\documentclass[a4paper,fleqn]{cas-dc}
\usepackage{amsmath,amssymb,amsfonts}
\usepackage{mathtools}
\usepackage{bm}
\usepackage{amsthm}

\usepackage{graphicx}
\usepackage{subcaption}
\usepackage{booktabs}
\usepackage{threeparttable}
\usepackage{tabularx}
\usepackage{multirow}
\usepackage{longtable}
\usepackage{tablefootnote}
\usepackage{float}
\usepackage{lscape}
\usepackage{xcolor,colortbl}
\usepackage{wrapfig}
\usepackage{stfloats}
\usepackage{placeins}
\usepackage{algorithm}
\usepackage[noend]{algpseudocode}
\usepackage{algcompatible}

\usepackage{textcomp}
\usepackage{xspace}
\usepackage{pifont}
\usepackage{calc}
\usepackage{fancyhdr}
\usepackage[shortlabels]{enumitem}
\usepackage{url}
\usepackage{varioref}
\usepackage{tikz}
\usepackage{pgfgantt}
\usepackage[normalem]{ulem}  
\usepackage{lipsum}          
\usepackage{ulem}
\usepackage{xcolor}
\usepackage{times}
\usepackage{soul}
\usepackage[numbers]{natbib}

{\bfseries}{\rmfamily}
{\bfseries}{\rmfamily}
{\bfseries}{\rmfamily}
\newtheorem{definition}{Definition}[section]
\newcommand{\hitsec}{\ensuremath {\texttt{HITFS}   {\xspace} }}
\newcommand{\borg}{\ensuremath {\texttt{BORG}   {\xspace} }}

\newcommand{\setup}{\ensuremath {\texttt{Setup} {\xspace}}}
\newcommand{\keyextract}{\ensuremath {\texttt{Extract} {\xspace}}}
\newcommand{\sign}{\ensuremath {\texttt{Sign} {\xspace}}}
\newcommand{\verify}{\ensuremath {\texttt{Verify} {\xspace}}}
\newcommand{\preprocess}{\ensuremath {\texttt{Preprocess} {\xspace}}}
\newcommand{\keygen}{\ensuremath {\texttt{KeyGen} {\xspace}}}
\newcommand{\mverify}{\ensuremath {\texttt{MVerify} {\xspace}}}
\newcommand{\pof}{\ensuremath {\texttt{PoF} {\xspace}}}
\newcommand{\fverify}{\ensuremath {\texttt{PoFVerify} {\xspace}}}

\newcommand{\amf}{\ensuremath {\mathit{AMF} {\xspace}}}

\newcommand{\bs}{\ensuremath {\mathit{BS} {\xspace}}}

\newcommand{\ue}{\ensuremath {\mathit{UE} {\xspace}}}

\newcommand{\zq}{\ensuremath \mathbb{Z}_{q}{\xspace}}
\newcommand{\id}{\ensuremath {\mathit{ID}}{\xspace}}

\newcommand{\sk}{\ensuremath {\mathit{sk}}{\xspace}}

\newcommand{\pk}{\ensuremath {\mathit{PK}}{\xspace}}
\newcommand{\ts}{\ensuremath {\mathit{TS}}{\xspace}}

\newcommand{\A}{$\mathcal{A}$~}

\newcommand{\C}{$\mathcal{C}$~}

\newcommand{\as}{\ensuremath {\leftarrow}{\xspace}}

\newcommand{\eufsidcmia}{\ensuremath {\textit{EUF}\mbox{-}\textit{sID}\mbox{-}\textit{CMIA}}{\xspace}}

\newcommand{\thpq}{\ensuremath {\texttt{ThPQ}{\xspace}}}
\newcommand{\aggregate}{\ensuremath {\texttt{Aggregate}{\xspace}}}

\newcommand{\PC}[1]{
	\vspace{2px}
	\noindent{\bf \IfEndWith{#1}{:}{#1}{#1:}}
}
\newcommand{\asrand}{\ensuremath {\xleftarrow{\$}}{\xspace}}

\newtheorem{theorem}{Theorem}

\def\tsc#1{\csdef{#1}{\textsc{\lowercase{#1}}\xspace}}
\tsc{WGM}
\tsc{QE}

\begin{document}
\let\WriteBookmarks\relax
\def\floatpagepagefraction{1}
\def\textpagefraction{.001}

\shorttitle{Future-Proofing Authentication Against Insecure Bootstrapping for 5G Networks: Feasibility, Resiliency, and Accountability}    

\shortauthors{Darzi et al.}  

\title [mode = title]{Future-Proofing Authentication Against Insecure Bootstrapping for 5G Networks: Feasibility, Resiliency, and Accountability}



\author[1]{Saleh Darzi}
\ead{salehdarzi@usf.edu}

\author[2]{Mirza Masfiqur Rahman}
\ead{rahman75@purdue.edu}

\author[3]{Imtiaz Karim}
\ead{imtiaz.karim@utdallas.edu}

\author[4]{Rouzbeh Behnia}
\ead{behnia@usf.edu}

\author[1]{Attila Altay Yavuz}
\ead{attilaayavuz@usf.edu}

\author[2]{Elisa Bertino}
\ead{bertino@purdue.edu}

\affiliation[1]{organization={Bellini College of Artificial Intelligence, Cybersecurity and Computing, University of South Florida},
            city={Tampa},
            postcode={33620},
            state={Florida},
            country={USA}}

\affiliation[2]{organization={Department of Computer Science, Purdue University},
            city={West Lafayette},
            postcode={47907},
            state={Indiana},
            country={USA}}

\affiliation[3]{organization={Department of Computer Science, University of Texas at Dallas},
            city={Richardson},
            postcode={75080},
            state={Texas},
            country={USA}}

\affiliation[4]{organization={School of Information Systems at University of South Florida},
            city={Tampa},
            postcode={33620},
            state={Florida},
            country={USA}}

\begin{abstract}
The 5G protocol lacks a robust base station (BS) authentication mechanism during the initial bootstrapping phase, leaving it susceptible to fake BSs, spoofed broadcasts, and large-scale manipulation of System Information Blocks (SIBs). Existing solutions incur high 
communication overhead, rely on centralized trust, and lack accountability and long-term breach resiliency. Given the inevitability of BS compromise and the severe impact of forged SIBs as the root of trust (e.g., fake alerts, tracking, false roaming), distributed trust, verifiable forgery detection, and audit logging are essential yet remain largely unexplored. These challenges are further amplified by the emergence of quantum-capable adversaries. 
While NIST Post-Quantum Cryptography (PQC) standards are widely viewed as a path toward long-term security, their feasibility under 5G's strict packet-size, latency, and broadcast constraints has not been systematically studied. This work presents, to our knowledge, the first comprehensive network-level performance characterization of integrating NIST-PQC standards and conventional digital signatures into 5G BS authentication, showing that direct PQC adoption is impractical due to excessive signature sizes, fragmentation, and protocol-level delays. To address these challenges, we propose BORG, a future-proof authentication framework based on a Hierarchical Identity-Based Threshold Signature with 
Fail-Stop (HITFS) properties. BORG distributes trust across multiple BSs via threshold signing, enables post-mortem verifiable forgery detection, and provides tamper-evident, PQ-secure audit logging, while maintaining compact signatures that fit within a single SIB1 packet without fragmentation and incurring minimal UE overhead, as validated through our real over-the-air 5G testbed implementation.
\end{abstract}

\begin{keywords}
5G Cellular Networks \sep Authentication \sep Network Performance Analysis \sep Transitional Post-Quantum Security
\end{keywords}

\maketitle

\section{Introduction} \label{sec:introduction} 
Despite advancements in next-generation cellular networks, the absence of a secure and efficient bootstrapping mechanism between User Equipments (UEs) and Base Stations (BSs) critically undermines the security of the 5G networks. The bootstrapping protocol enables UEs to connect to BSs and access the core network. However, during the initial Radio Resource Control (RRC) connection, the UE selects a BS based solely on signal strength and broadcast parameters~\cite{wuthier2024fake}, without any verified BS identity and authentication. This absence of a robust or standardized BS authentication leads to severe security attacks, such as fake base stations, phishing, and spoofed emergency alerts~\cite{hussain2019insecure, cao2019survey}. 

\subsection{Prior Works on BS Authentication} 
To address the absence of BS authentication, prior efforts have pursued several broad directions. Certificate-based approaches adapt Public Key Infrastructure (PKI) frameworks to attach digital signatures and certificate chains to SIB messages, providing verifiable BS identity at the cost of significant communication overhead~\cite{ross2024fixing, lee2009extended, zheng1996authentication, FBSSpecifications}. To reduce this overhead, certificate-free designs based on Identity-Based Signatures (IBS) derive BS and AMF keys hierarchically from a master key pre-installed in the USIM, eliminating certificate transmission and achieving more compact footprints~\cite{singla2021look, sun20255g}. Token-based and symmetric schemes offer lightweight pre-authentication without asymmetric overhead but do not protect SIB content itself~\cite{lotto2023baron, perrig2003tesla}. More recent efforts have explored threshold signatures to distribute signing responsibility across multiple BSs~\cite{sengupta2024fast, vikhrova2022multi}, and hybrid constructions have begun addressing long-term security for cellular authentication~\cite{vuppala2023post, ko20255g, scalise2024applied}. 

\looseness-1 Building on these solutions, various optimizations, such as efficient certificate delivery~\cite{gao2021evaluating, al2025innovative}, online-offline signatures~\cite{sengupta2024fast}, and outsourced computation via auxiliary entities~\cite{ramadan2020identity}, have been proposed to reduce signing and certificate overheads~\cite{hussain2019insecure}, with recent works further advancing efficiency through certificate-free designs based on hierarchical IBSs~\cite{yu2024protecting}. Although recent efforts from 3GPP consider protecting the unicast RRC messages, initial bootstrapping messages like SIB1 are still unprotected. In particular, TS 38.331 Annex B.1 states that even after Access Stratum (AS) security activation, SIB1 can be sent without integrity protection and ciphering~\cite{RRCSpec}. Below, we outline critical research gaps and the security requirements for 5G BS authentication. 
\textit{(i)}~\ul{\textit{Lack of Efficient Distributed Authentication in 5G:}} The security of 5G bootstrapping currently depends on a single BS. However, modern cellular deployments increasingly rely on heterogeneous, densely deployed, and physically exposed infrastructures (e.g., small cells, femtocells, and O-RAN units~\cite{polese2023understanding, al2022software}), where software modification and physical access significantly expand the attack surface~\cite{rupprecht2019breaking}. Thus, individual BSs are more susceptible to compromise through rooting, tampering, or other software exploits~\cite{janzen2024oh}. This creates a structural single point of failure: compromise of even one BS can break authentication guarantees for all UEs under its coverage, enabling fake-BS attacks, SIB manipulation, tracking, and DoS. The KT femtocell attack~\cite{kt-femtocell} demonstrates this risk in practice, showing that inexpensive, user-provisioned hardware can be rooted and used to inject malicious SIB1 or control uplink/downlink communication, making single-BS trust inadequate. 

\looseness-1 Furthermore, multi-connectivity, where UEs interact with Master and Secondary gNBs (EN-DC, NE-DC), carrier-aggregation cells, CoMP clusters, and Xn-coordinated nodes, is already the operational norm in 5G~\cite{hussain2019insecure}. Leveraging multiple BSs for authentication therefore aligns with existing architectural practices and removes the fundamental mismatch between multi-BS communication and single-BS trust. Even when a single logical BS serves a region, modern 5G RAN architectures increasingly virtualize and disaggregate BS functions using software-defined and cloud-native principles, allowing protocol and security operations to be executed across multiple virtualized components and compute nodes rather than a single physical entity. Thus, an effective authentication solution must distribute trust across multiple BSs or instances thereof. Threshold signatures (e.g.,~\cite{komlo2021frost}) enable this by requiring a quorum of independent stations to collaboratively authenticate broadcast system information, thereby eliminating the single point of failure, reducing BS/key compromise risks~\cite{de2024performance}, strengthening attacker resistance, and aligning with existing 5G architectural patterns. Despite this need, only a few efforts~\cite{vikhrova2022multi,sengupta2024fast} have explored threshold signatures for 5G authentication. Moreover, UEs resource constraints, limited packet sizes, and frequent broadcasts, especially in distributed BS settings, make efficiency crucial. Thus, the designed authentication must minimize signature, communication, and storage overhead while ensuring fast signing and verification.


\textit{(ii)}~\ul{\textit{Lack of Accountability, Breach Resiliency, and Long-Term Security in 5G Authentication:}} Given that the $SIB$ message serves as the root of trust in 5G bootstrapping, any forged signature effectively gives an attacker control over a UE’s view of the network and enables impactful attacks such as large-scale impersonation, MITM, fake alerts, tracking, stealthy DoS~\cite{darzi2024counter}, and false roaming. Because BS compromises are ultimately inevitable—whether through physical access, software exploits, advanced processing power, or insider threats—prevention alone is insufficient. A robust authentication design must support reliable detectability and distinguish malicious forgeries from benign failures rather than relying solely on idealized security assumptions.

This need becomes even more critical when viewed alongside real-world intrusion behavior: industry evidence shows that security breaches often persist undetected for long periods. Data breach reports (e.g., IBM~\cite{ibmreport}, Verizon \cite{verizonreport}) estimate an average dwell time of roughly 204 days before detection, giving attackers months to escalate privileges, manipulate system behavior, and stage follow-on operations without being noticed. In the case of 5G, the recent SK Telecom data breach shows the attackers were in the system and waited to affect 28 servers before doing the data breach~\cite{sk-telecom}. Similarly, a forged SIB or compromised BS could remain active for extended periods, silently influencing mobility, routing, and service selection. Once an adversary injects a malicious broadcast, it can repeatedly mislead UEs or force attachment to rogue cells, enabling sustained exploitation. Without a tamper-evident mechanism to detect such manipulation, these attacks can remain invisible indefinitely.

Moreover, in security-critical infrastructures, such as 5G, 
the ability to prove that a breach occurred is as important as preventing one. Yet current 5G architecture provides no cryptographically enforced provenance of broadcast signatures, no mechanism to identify compromised signers, and no way to produce irrefutable evidence of misbehavior. This absence prevents operators from halting further exploitation, guiding remediation, or performing forensics after an attack. Verifiable forgery detection and tamper-evident logging are therefore essential for post-mortem analysis, accountability, and long-term integrity, particularly when logs from multiple BSs or core entities can be cross-validated to detect inconsistencies, forks, or replay attempts. However, existing RAN logs are purely operational, lack cryptographic binding, and are susceptible to modification by compromised software or attackers with elevated privileges, leaving no trustworthy link between broadcast signatures and recorded events. Consequently, 5G authentication faces a pressing challenge: the absence of a distributed and verifiable audit logging mechanism that provides long-term security and non-repudiation.

\textit{(iii)}~\ul{\textit{Lack of Systematic Feasibility 
Analysis for Long-Term Security:}} These challenges are further amplified once large-scale quantum computers emerge, as quantum algorithms are expected to undermine the computational hardness assumption of the classical public-key primitives. Combined with the inevitability of BS compromise and practical implementation failures (e.g., side-channel attacks), this means that signature forgeries will occur over time, reinforcing the need to future-proof 5G authentication. This urgency is also reflected in global efforts initiated by the National Institute of Standards and Technology (NIST) on standardizing Post-Quantum Cryptography (PQC)~\cite{darzi2023envisioning} and the migration guidelines of the European Telecommunications Standards Institute (ETSI) and IEEE Standards Association, which emphasize preparing communication systems for the Post Quantum (PQ) era~\cite{etsi}. While preliminary integration of NIST-PQC algorithms into various network protocols, such as TLS~\cite{DBLP:conf/ndss/SikeridisKD20} and PQ-WireGuard, has begun, these efforts reveal significant trade-offs, highlighting the substantial overhead and limited practicality of current PQ schemes in constrained, latency-sensitive mobile environments. For 5G BS authentication, these limitations are even more restrictive due to strict size, timing, and broadcast constraints; for example, initial frame synchronization messages like SIB1 are limited to $372$ bytes and are transmitted periodically with a delay around $160$ ms (discussed in detail in Section~\ref{sec:PerformanceEvaluation}). Yet the feasibility of integrating NIST-PQC signatures into 5G bootstrapping, particularly in distributed or hierarchical settings with protocol-level intricacies, has not been comprehensively analyzed, leaving a critical open question in the evolution of secure 5G authentication.

\subsection{Our Contributions} \label{subsec:Contribution}
To address these challenges, we begin by evaluating the feasibility of integrating NIST-PQC standards into 5G BS authentication through a protocol-level performance analysis. Our findings highlight severe performance bottlenecks that arise when directly applying NIST-PQC signatures to 5G BS authentication. Motivated by these limitations, we present $\borg$, a future-proof authentication framework that provides efficient threshold (distributed) signing with post-mortem (PM) forgery detection and a verifiable auditing mechanism. Our key contributions are as follows:

\textit{\textbf{(i)}}~\ul{\textit{\textbf{Analysis of NIST-PQC Standards and Conventional Alternatives:}}} 
To our knowledge, this is the first in-depth evaluation of the NIST-PQC scheme adoption in the context of 5G BS authentication. Focusing on SIB1 as the critical broadcast message in the bootstrapping process, we show that directly applying NIST-PQC signatures is impractical, as detailed in Section~\ref{sec:feasibility}. Specifically, the large signature and certificate sizes of NIST's primary (lattice-based) PQC standard $\textit{ML\mbox{-}DSA}$ \cite{dang2024module} impose a $12276$-byte communication overhead, requiring fragmentation and causing significant 5G packet delays. This results in a total end-to-end delay of $5282$~ms, which is incompatible with the real-time requirements of 5G BS communication. We also assess conventional hierarchical IBS schemes to broaden the performance profile; while offering good performance~\cite{singla2021look}, they lack support for distributed trust, accountability, and long-term forgery detection.


\looseness-1 
\textit{\textbf{(ii)}}~\ul{\textit{\textbf{$\borg$: An Efficient, Distributed, and Accountable Authentication Framework:}}} Given the infeasibility of NIST-PQC signatures for 5G BS authentication, we design $\borg$, a future-proof authentication framework based on a Hierarchical Identity-Based Threshold Signature with Fail-Stop (HITFS) property. {\em (1)} $\borg$ distributes trust across multiple BSs via threshold signatures, ensuring resiliency even when some BSs are compromised or misbehave. {\em (2)} It provides conventionally secure authentication with post-mortem, verifiable forgery detection through a fail-stop mechanism. Even if the underlying classical hardness assumptions (e.g., discrete logarithms) fail in the future, honest and computationally bounded signers can identify and prove forgeries, including those produced by a quantum-capable adversary, thereby ensuring strong accountability and long-term security. This forgery-detection capability is reinforced through a distributed audit logging mechanism, ensured via PQ-secure threshold signatures, ensuring tamper-evident records, non-repudiation, and long-term system resiliency.  {\em (3)} Despite supporting distributed trust and accountability, $\borg$ maintains compact signatures, low communication overhead, minimal UE computation, and reduced end-to-end delay, eliminating the need for fragmentation. Compared to the existing conventional-secure alternatives like $\textit{Schnorr\mbox{-}HIBS}$ (see TABLE~\ref{tab:PerformanceComparison2} \& \ref{tab:PerformanceComparison3}), $\borg$ achieves similar overhead while additionally offering distributed authentication, accountability, and forgery detection.

\textit{\textbf{(iii)}} \ul{\textit{\textbf{Open-Sourced Evaluation Framework:}}} We fully implemented $\borg$ by incorporating it into a real 5G testbed in srsRAN and then conducted an extensive performance evaluation against existing authentication schemes. Tested with over-the-air 5G communication, $\borg$ demonstrates practical deployability with low computational and communication overhead. Compared to aggregate signature schemes such as $\textit{BLS}$, $\borg$ achieves faster signing and lower end-to-end delay with reduced communication overhead in 5G settings. Relative to IBS schemes~\cite{singla2021look}, $\borg$ attains similar runtime while also providing distributed authentication, audit logging, and post-mortem forgery detection. Moreover, $\borg$ is up to three orders of magnitude faster and incurs $85\times$ less communication overhead than NIST-PQC's $\textit{ML\mbox{-}DSA}$~\cite{dang2024module}. These results highlight $\borg$ as a compact, efficient, and future-proof solution for 5G bootstrapping authentication. To support reproducibility, we publicly release the complete source code of $\borg$\footnote{\fbox{\href{https://github.com/TheSalehDarzi/BORG-Scheme/tree/main/BORG} {\texttt{github.com/TheSalehDarzi/BORG-Scheme/tree/main/BORG}}}}, along with its over-the-air 5G testbed implementation\footnote{\fbox{\href{https://github.com/TheSalehDarzi/BORG-Scheme/tree/main/OTA} {\texttt{github.com/TheSalehDarzi/BORG-Scheme/tree/main/OTA}}}}.

\vspace{-2mm}
\subsection{Outline} \label{subsec:outline}
The remainder of this paper is organized as follows. Section~\ref{sec:preliminaries} presents the system model, notation, and cryptographic building blocks underlying $\borg$. Section~\ref{sec:threatandsecuritymodels} formalizes the threat model, scope, and security definitions. Section~\ref{sec:feasibility} analyzes the feasibility of integrating NIST-PQC standards into 5G BS authentication, demonstrating the impracticality of direct adoption. Section~\ref{sec:solution} presents the $\borg$ framework, including the proposed scheme, security analysis, and full 5G protocol instantiation. 
Section~\ref{sec:PerformanceEvaluation} provides a comprehensive performance evaluation against PQC and conventional baselines on a real over-the-air 5G testbed. Section~\ref{sec:RelatedWork} surveys related work. Section~\ref{sec:conclusion} concludes the paper and outlines future directions.
\vspace{-3mm}
\section{Preliminaries and Building Blocks} \label{sec:preliminaries}
This section outlines the notation, network architecture, and the cryptographic building blocks. 

\noindent \textbf{Notations}: The symbol $||$ denotes concatenation, and $\cdot$ denotes multiplication. 
For two primes $p$ and $q$, let $\mathbb{Z}_q$ be the finite field of integers modulo $q$, and let $\mathbb{G}$ be a cyclic group of prime order $p$ with generator $g$. We define two cryptographically secure hash functions: $H_1: \{0,1\}^\ast \rightarrow \mathbb{Z}_q$ and $H_2: \{0,1\}^{*} \rightarrow \mathbb{Z}_q$. $x \xleftarrow{\$} \mathcal{S}$ indicates that $x$ is sampled uniformly at random from the set $\mathcal{S}$. Vectors are denoted by $\vec{x}$, and $\{x_i\}_{i=1}^{n} = \{x_1, x_2, \dots, x_n\}$ represents a set of $n$ elements. Finally, $\sk$, $\pk$, and $\id$ refer to the secret key, public key, and the identity of an entity (e.g., MAC address), respectively. A list of acronyms is provided in Appendix~\ref{subsec:acronmys}.

\subsection{System Model: 5G Cellular Network}
\label{subsec:5G} 
\subsubsection{Network Components}
The 5G network consists of three main entities~\cite{fourati2021comprehensive}: \vspace{-2mm}

\begin{itemize}[leftmargin=*]
    \item \textit{\textbf{5G Network Core:}} This entity serves as the central management of the cellular network, responsible for service delivery, session management, policy control, data handling, and security enforcement while integrating multiple network functions. One of the crucial components of the network core is the Access and Mobility Management Function (AMF), which is most relevant to our work.\vspace{-2mm}
    \item \textit{\textbf{User Equipment (UE):}} Located at the network edge, a UE refers to a cellular device (e.g., smartphone or IoT) subscribed to the network. Each UE is registered and equipped with a Universal Subscriber Identity Module (USIM) issued by the network authorities. It uses a unique identifier for communication, connection establishment, and access to network services.  \vspace{-2mm}
    \item \textit{\textbf{Radio Access Network (RAN):}} This network, comprising BSs (gNB) and UEs, manages radio transmissions, traffic, data exchange, and user service requests. Our work focuses on this component, where bootstrapping and system information messages are periodically broadcast. As these messages are neither encrypted nor signed, they are susceptible to adversarial manipulation. The UE initiates service requests procedures based on the content of these broadcasts and the type of service required.  \vspace{-3mm}  
\end{itemize}

\begin{figure}
	\centering
	\includegraphics[scale=0.4, trim = 0cm 18cm 0cm 0cm, clip]{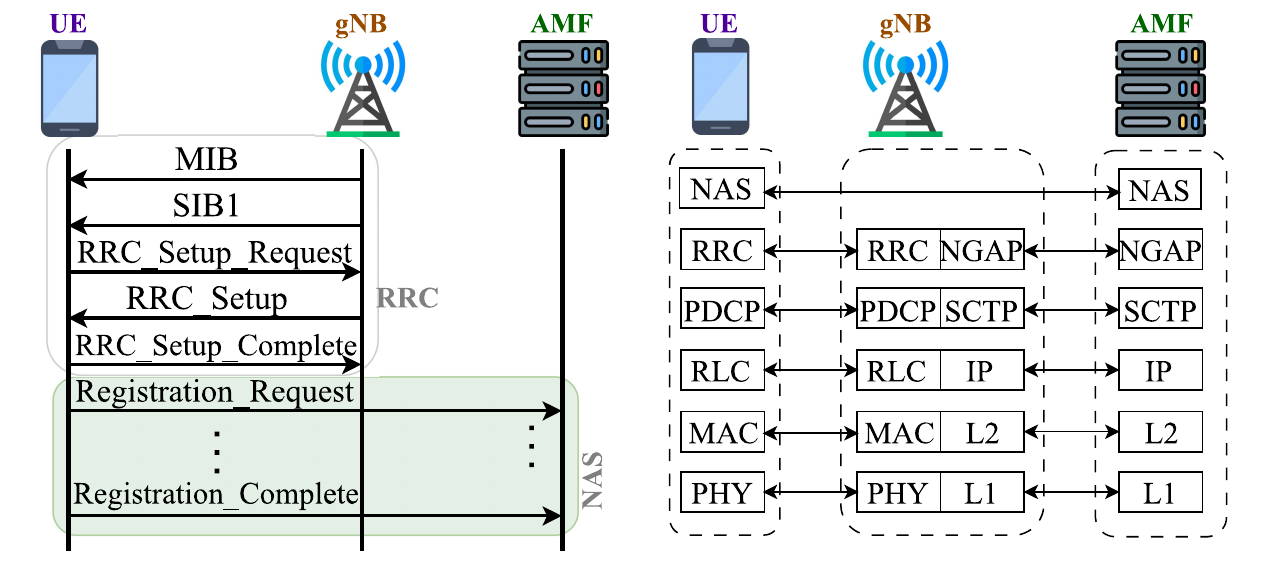}
	\caption{\small Initial 5G network connection setup and protocol stack.} \vspace{-6mm}
	\label{fig:timing}
\end{figure}

\subsubsection{Initial $\boldsymbol{\bs}$-$\boldsymbol{\ue}$ Communication} 
In 5G protocol stack (Figure~\ref{fig:timing}), the topmost layer in the BS is Radio Resource Control (RRC). For UE and AMF, the Non-Access Stratum (NAS) layer is stacked over the RRC layer. Only the master public key ($\pk_{\id_0}$) is securely embedded in the USIM and is assumed to be publicly verifiable. Private keys for AMF and BSs are derived from the master secret key ($\sk_{\id_0}$) and distributed through secure channels. 
For initial bootstrapping, the BS broadcasts the System Information (SI)—an RRC message—to the UE, announcing its configuration parameters. SI consists of the Master Information Block (MIB) and the System Information Block (SIB). The MIB, a short message, assists in decoding the first SIB: SIB1.  
Broadcast over the Downlink Shared Channel (DL-SCH), SIB1 contains scheduling and availability information for other SIB messages. As per 3GPP specifications~\cite{RRCSpec}, SIB1 has a maximum size of 372 bytes and is transmitted periodically every $160$~ms, with repeated broadcasts allowed.  
Fig.~\ref{fig:sib1} illustrates the structure of SIB1. Some fields are always present, while others are conditional or optional. The \textit{Cell Selection Info} field provides signal quality metrics, while \textit{Cell Access Related Info} includes Public Land Mobile Network (PLMN) identifiers and cell access status. Optional fields like \textit{IMS-Emergency Support} indicate support for emergency services in limited service mode. For detailed field descriptions, see~\cite{RRCSpec}. After receiving SIB1, the UE initiates the RRC setup. Upon successful RRC connection, the UE initiates NAS registration with the AMF.  
Several additional SIB messages (SIB2–SIB21) are transmitted over DL-SCH in periodic windows, each serving specific functions. For instance, SIB3 provides NR \textit{intra}-frequency neighbor cell lists and reselection areas, SIB4 conveys \textit{inter}-frequency equivalents, SIB9 delivers GPS and UTC time, and SIB15 carries disaster roaming configurations.

\begin{figure*}[ht!]
	\centering
	\includegraphics[scale=0.8, trim = 0cm 24cm 0cm 0cm, clip]{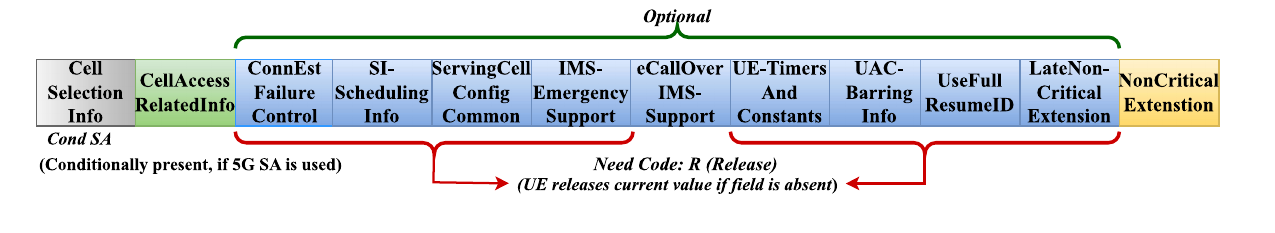} \vspace{-6mm}
	\caption{SIB1 message structure in 5G.}\vspace{-5mm}
	\label{fig:sib1}
\end{figure*}

\subsection{Building Blocks} \label{subsec:cryptoprimitives}

\looseness-1
\noindent\textbf{Hierarchical Identity-Based Threshold Signature Scheme with Fail-Stop property (HITFS).} 
Our proposed $\borg$ framework realizes a Hierarchical Identity-Based Threshold Signature scheme with Fail-Stop property (HITFS), which harnesses a Hierarchical IBS (HIBS)\cite{chow2004secure, singla2021look} and FROST~\cite{komlo2021frost}. In the HIBS model, keys are derived hierarchically, where each level’s keys are generated from its parent, binding identities directly to signing keys without the need for a trusted certificate authority. This structure supports efficient identity validation and key expiration verification using a compact master key~\cite{ramadan2020identity, chow2004secure, galindo2009schnorr}. By incorporating threshold cryptography, any $t$ out of $n$ authorized signers can collaboratively produce a valid signature without reconstructing the group’s secret key, while fewer than $t$ participants are cryptographically incapable of forging a signature. This design is particularly well-suited for distributed and fault-tolerant signing, as each participant only holds a share of the signing key. As long as the number of colluding parties remains below the threshold $t$, unauthorized signing remains infeasible~\cite{ergezer2020survey}. In addition, $\borg$ integrates an FS security mechanism~\cite{susilo1999fail, yaksetig2024extremely}, which leverages the second pre-image resistance of cryptographic hash functions to enable post-mortem forgery detection against quantum adversaries. While FS operates similarly to standard signatures under conventional security assumptions, it uniquely allows a signer to prove that a forgery has occurred if those assumptions are violated. This serves as a breach resiliency mechanism, halting further key usage and providing cryptographic evidence to absolve the signer of liability. 

\vspace{-1mm}
\begin{definition} \label{def:HIBS}
A hierarchical identity-based threshold signature scheme with fail-stop property is a 7-tuple algorithm as shown below: \vspace{-2mm}
    \begin{itemize}[leftmargin=*]
	\item[-] ${\underline{(\sk_{\id_0}, \pk_{\id_0}, params)\as \hitsec.\setup(1^\kappa)}}$:~Given~the security parameter $\kappa$, it outputs the master secret and public keys $(\sk_{\id_0}, \pk_{\id_0})$ and the system parameters, $params$, which is an implicit input to all the following algorithms.\vspace{-2mm}
	\item[-] ${\underline{(\{\sk_{\id_{k,i}}\}_{i=1}^{n}, {\vec{Q}}_{\id_k}) \as \hitsec.\keyextract({\vec{\id}}_k, {\vec{Q}}_{\id_{(k-1)}}}}$, ${\underline{\sk_{\id_{(k-1)}})}}$:~Given the identity vector at level $k$~${\vec{\id}}_k=$ $(\id_1,\id_2, \dots, \id_k)$, the algorithm extracts the secret key of $\id_k$ using the public key vector ${\vec{Q}}_{\id_{(k-1)}}$ and the secret key $\sk_{\id_{(k-1)}}$ from level $k-1$. It then outputs the secret key shares for each participant ($\{\sk_{\id_{k,i}}\}_{i=1}^{n}$) and computes the corresponding group public key values ${\vec{Q}}_{\id_k} = (Q_{\id_1}, Q_{\id_2}, \dots, Q_{\id_k})$.\vspace{-2mm}
    \item[-] $\underline{\mathcal{L}_i\as \hitsec.\preprocess(J)}$:~Given the predeterm- ined number of messages to be signed $J$, it returns the commitment values for all participants ${i\in [1,n]}$ in a list $\mathcal{L}_{i} \as (i, \{E_{i,j}\}_{j=1}^{J}, \{D_{i,j}\}_{j=1}^{J})$. \vspace{-2mm}
    \item[-] $\underline{\sigma_{k,j} \as \hitsec.\sign(m_j, \mathcal{L}_{i}, \{\sk_{k,i}\}_{i=1}^{\beta})}$:~Given a mes- sage $m_j$ with index $j$, commitmetn values $\mathcal{L}_i$, and $\beta \in [t,n]$ participating signers' secret keys ($\{\sk_{k,i}\}_{i=1}^{\beta}$), it returns a signature $\sigma_{k,j}$ for signers at level $k$. \vspace{-2mm}
    \item[-] $\underline{\{0, 1\} \as \hitsec.\mverify(m_j, {\vec{\id}}_k,{\vec{Q}}_{\id_k}, \sigma_{k,j})}$:~It re- turns $1$ if the signature $\sigma_{k,j}$ on message $m_j$ is valid with respect to the identity vector ${\vec{\id}}_k$ and public key vector ${\vec{Q}}_{\id_k}$, and $0$ otherwise. \vspace{-2mm}
    \item[-] $\underline{\pi \as \hitsec.\pof(\{\hat{e}_{i,\mathrm{j}}\}_{i=1}^{\beta}, \{\hat{d}_{i,\mathrm{j}}\}_{i=1}^{\beta}, m, \sigma_k', hist)}$:~Given the message-signature pair $(m, \sigma_k')$, random commitment values of the signing participants, and the history of previous signatures ($hist$), it outputs $\pi$, a proof of forgery if $\sigma_k'$ is forged; otherwise, it returns “\mbox{Not A Forgery}”. \vspace{-2mm}
    \item[-] $\underline{\{0, 1\} \as \hitsec.\fverify(\alpha_{k}, \sk_{k-1}, Q_{\id_k}, m, \sigma_k', \pi)}$: On the selected random input $\alpha_{k}$, public $Q_{\id_k}$, secret key $\sk_{k-1}$, message $m$, signature $\sigma_k'$, and $\pi$, it returns $1$ if the proof of forgery is valid, otherwise, $0$. 
	\end{itemize}
\end{definition}

\begin{definition} \label{def:DLP}
\looseness-1 Discrete Logarithm Problem (DLP)~\cite{johnson2001elliptic}: Let $\mathbb{G}$ be the finite cyclic group with generator $g$, given $g \in \mathbb{G}$, $h \in \mathbb{G}$, and $h = g^x \mod p$ with some unknown $x \in \mathbb{Z}_q$, the $(EC)DLP$ requires computing $x = \log_g h\mod p$. 
\end{definition}

\noindent \textbf{PQ Threshold Digital Signature Scheme}. We employ a PQ-secure threshold signature scheme ($\thpq$) for distributed audit logging, eliminating single points of failure. $\thpq$ is critical for audit logging and subsequent forgery detection in our future-proofed defense against quantum-capable adversaries. It distributes the audit key across multiple BSs, allowing any $t$ of them to jointly sign, while preventing forgery by up to $(t{-}1)$ compromised nodes~\cite{boschini2024ringtail}.  
The scheme consists of: {\em (i)} $\thpq.\keygen(1^\kappa, t, n)$: generates a global public key $\pk$ and a set of $n$ secret key shares $(\sk_1,\dots, \sk_n)$ from the security parameter $\kappa$ and threshold $t$ out of $n$; {\em (ii)} $\thpq.\sign(\sk_i, m)$: each signer $i$ for ${i=1,\dots,t}$, uses $\sk_i$ to produce a signature share $\sigma_i$ on the message $m$; {\em (iii)} $\sigma \as \thpq.\aggregate(\{\sigma_i\}_{i=1}^{t})$ aggregates $t$ signature shares into a valid signature $\sigma$. {\em (iv)} $\{0,1\}\as\thpq.\verify(\pk, m, \sigma)$ verifies $\sigma$ on message $m$ using $\pk$. For further details, see~\cite{boschini2024ringtail}.

\section{Threat and Security Models} \label{sec:threatandsecuritymodels}
This section outlines the threat model and the scope of our solution followed by the security model that underpin the formal security proof of $\borg$.

\subsection{Threat Model and Scope} \label{subsec:threatmodel} 
We consider a probabilistic polynomial-time (PPT) adversary with full control over the wireless medium. The adversary can eavesdrop on all broadcast messages, inject, modify, or replay forged $SIB$ messages, and impersonate legitimate base stations (gNBs) to mislead UEs. Additionally, the adversary may corrupt up to ${(t-1)}$ BSs, gaining access to their secret keys and internal states to craft forgeries. The adversary is thus capable of performing three attack vectors commonly exploited in cellular networks, as captured in our threat model and illustrated in Fig.~\ref{fig:adv_models}, and detailed below:\vspace{-2mm}


\begin{itemize}[leftmargin=*]
    \item \textbf{Fake Base Stations (FBSs).} These attacks~\cite{ltefuzz} are carried out by luring the victim UE to connect to an FBS that spoofs legitimate BSs. Once connected, attackers can launch multi-phase attacks that exploit vulnerabilities in subsequent protocol stages~\cite{mubasshir2025gottadetectemall}. \vspace{-2mm}
    
    \item \textbf{Key Compromise Scenarios.} We account for active adversaries capable of compromising BSs to extract signing keys~\cite{de2024performance}, forge signatures, and impersonate legitimate BSs during 5G bootstrapping~\cite{ransacked}. While some BSs may be compromised, we assume at least $t$ out of $n$ remain uncompromised. This is a practical assumption, as a majority compromise would indicate that the entire network is no longer trustworthy. To support this, BS hardening techniques such as advanced intrusion detection and secure configuration practices can be applied~\cite{erricson}. \vspace{-2mm}
    
    \item \looseness-1 \textbf{MiTM Attacker.} An MiTM attacker impersonates a BS to a victim UE and vice versa, 
    enabling interception, modification, or replay of messages. This is possible when traffic is not protected (e.g., digitally signed)~\cite{rupprecht2019breaking}. \vspace{-2mm}
    
\end{itemize}

Our threat model also captures provable detection of quantum-capable adversaries with the potential to break conventional signatures~\cite{mitchell2020impact}. While our scheme does not provide real-time PQ security, it enables post-mortem (PM) forgery detection via a fail-stop (FS) mechanism: computationally bounded signers (i.e., BSs) can identify and prove forgeries once the underlying assumptions are broken~\cite{pfitzmann1991fail, boschini2024s}. This FS mechanism halts the system upon security breaks, minimizing damage and further exploitation. Since forgery detection relies on the integrity of audit logs, we incorporate a distributed audit logging system secured with PQ threshold signatures to support post-mortem detection (see Section~\ref{subsubsec:5}).



\looseness-1 \noindent \textbf{Scope.} Our objective is to design an authentication framework for 5G UE-BS communication. Since SIB1 conveys critical RAN information and the broadcast schedule for subsequent SIB messages, its authentication ensures that devices receive legitimate access details and scheduling. We therefore identify SIB1 as the most essential message to protect, and implement $\borg$ primarily for its authentication. However, we identify that our mechanism is directly deployable for other SIB messages.  Our scope excludes authentication and key agreement procedures (e.g., 5G-AKA~\cite{rossi2024enhancing, damir2022beyond}), as $\borg$ targets the broadcast bootstrapping plane (SIB1 authenticity and BS legitimacy) that precedes NAS/AKA, while 5G-AKA provides UE–core mutual authentication and session key establishment. Hence, $\borg$ is orthogonal and complementary to AKA: it does not replace key agreement but ensures that the UE only initiates AKA with a verified BS.   
In practice, the Core Key Generator (Used in $\borg$) role can be realized as a logical function co-located with existing key-management entities, and the required master public key can be provisioned through standard USIM/eSIM updates. Also, $\borg$ can be integrated with the AKA procedure through a policy gate, ensuring that the AKA process is initiated only after a fresh $\borg$-verified SIB1 has been received. This linkage prevents fake BS–driven bootstrapping while preserving the cryptographic structure of AKA. Similarly, side-channel attacks, including physical key extraction, are also out of scope. Additionally, our framework does not address UE-to-BS privacy, denial of service, passive eavesdropping, jamming, overshadowing, or other physical-layer attacks, which require separate, orthogonal defenses on other layers of 5G such as physical-layer encryption or anti-jamming mechanisms.

\begin{figure*}
\centering
\begin{subfigure}{.25\textwidth}
  \centering
  \includegraphics[width=0.8\linewidth]{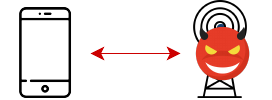}
  \caption{\small Fake Base Stations}
  \label{fig:sub1}
\end{subfigure}%
\begin{subfigure}{.25\textwidth}
  \centering
  \includegraphics[width=0.75\linewidth,trim = 0cm 19.5cm 0cm 0cm, clip]{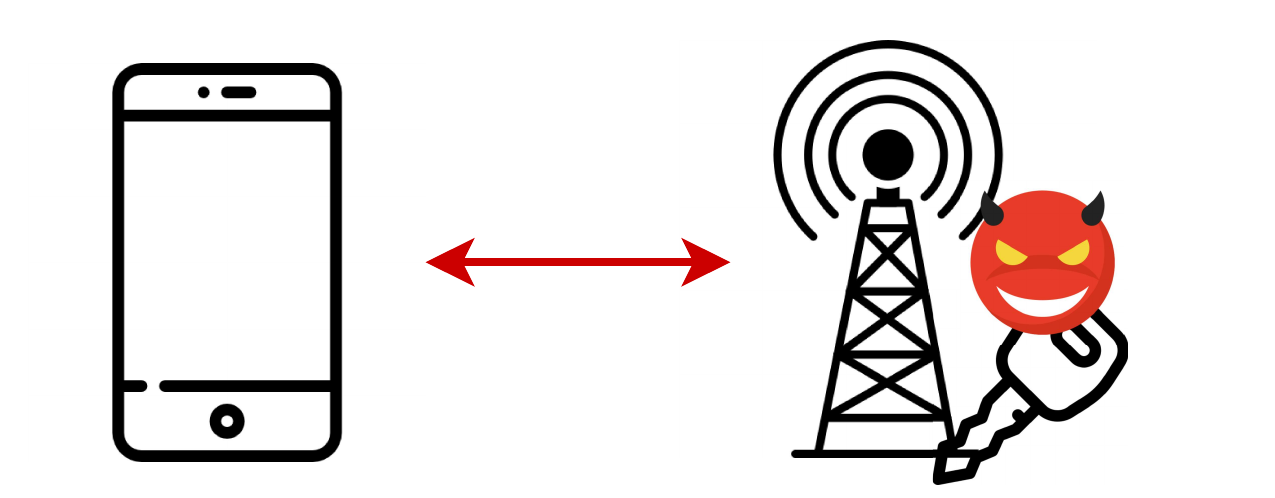}
  \caption{\small Compromised Key}
  \label{fig:sub4}
\end{subfigure}%
\begin{subfigure}{.25\textwidth}
  \centering
  \includegraphics[width=1\linewidth]{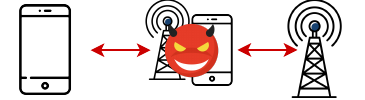}
  \caption{\small MiTM Attacker}
  \label{fig:sub3}
\end{subfigure}%
\begin{subfigure}{.25\textwidth}
  \centering
  \includegraphics[width=0.9\linewidth]{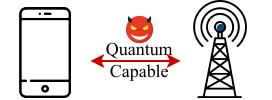}
  \caption{\small Quantum Adversary}
  \label{fig:sub4}
\end{subfigure}
\caption{Outline of Our Threat Models. 
}\vspace{-6mm}
\label{fig:adv_models}
\end{figure*}


\subsection{Security Model} \label{subsec:securitymodel} 
Following the hierarchical IBS~\cite{schnorr1990efficient, galindo2009schnorr} and Schnorr-based (threshold) signature security 
models~\cite{komlo2021frost, singla2021look}, we define the Existential Unforgeability under a selective-ID, adaptive Chosen Message-and-ID Attack (EUF-sID-CMIA) for a $\hitsec$ scheme through a game between the adversary \A and challenger $\mathcal{C}$. The adversary \A controls fewer than $t$ signing participants and has access to the following oracles:  
{\em (i)} Key Extraction Oracle $\mathcal{O}_{E}$: Given a user $\id$ at level $k$ with $n$ users, it returns the secret key shares $\{\sk_{ID_{k,i}}\}_{i=1}^{n}$.  
{\em (ii)} Preprocessing Oracle $\mathcal{O}_{P}$: On input signing round $j$ and user identity $\id$, it provides the commitment values for signing.  
{\em (iii)} Signing Oracle $\mathcal{O}_{S}$: Given a message $m$ and user $\id$, it executes the signing procedure and returns a valid signature $\sigma$.  
{\em (iv)} Random Oracles $\mathcal{O}_{H_1}$ and $\mathcal{O}_{H_2}$: Queries to hash functions $H_1$ and $H_2$ are modeled as interactions with a random oracle, an idealized black-box that returns truly random outputs for each unique query while maintaining consistency across repeated inputs.

\begin{definition} \label{def:Experiment1}
The EUF-sID-CMIA security experiment $\mathit{Expt}_{\hitsec}^{\eufsidcmia}$ for a $\hitsec$ signature scheme is defined as follows:\vspace{-1mm}
\begin{itemize}[leftmargin=*]
    \item[-] \C runs $\hitsec.\setup(1^\kappa)$ and returns the $\pk_{\id_0}$ and the public parameters to the adversary $\mathcal{A}$. \vspace{-2mm}
    \item[-] $(m^\ast, {\vec{\id}}^\ast_k, {\vec{Q}}_{\id^\ast_k}, \sigma^\ast) \as$ \A$^{\mathcal{O}_{\texttt{E}}, \mathcal{O}_{P}, \mathcal{O}_{S}}(\pk_{\id_0}, params)$ 
\end{itemize}
\end{definition}
\vspace{-1mm}
\A wins the experiment if the forged signature passes verification (${1 \as \mverify(m^\ast, {\vec{\id}}^\ast_k, {\vec{Q}}_{\id^\ast_k}, \sigma^\ast)}$) while satisfying the following conditions:  
{\em (i)} The target $\id^\ast$ or any of its prefixes was not queried to $\mathcal{O}_{E}$.  
{\em (ii)} The commitment values used for signing were not queried to the preprocessing oracle $\mathcal{O}_{P}$.  
{\em (iii)} The message-and-ID pair ${(m^\ast, {\vec{\id}}')}$, where ${\vec{\id}}'$ is a prefix of $\id^\ast$, was not queried to $\mathcal{O}_{S}$.  
{\em (iv)} The $\pof(.)$ or hash queries $H_1$ on the secret random value $\alpha_i$ (for $i \in \{0, k\}$) were not invoked during the security experiment. The forger's advantage in winning the game is defined as $Pr[\mathit{Expt}^{\eufsidcmia}_{\hitsec}$($\mathcal{A}$) $= 1]$.




Following the principles of fail-stop signature schemes \cite{pfitzmann1991fail, boschini2024s}, we formalize the security of a $\hitsec$ scheme through the following properties: {\em (i) Signer-Side Security}, which ensures that a quantum-capable adversary controlling fewer than $t$ signers cannot produce an undetectable forgery; {\em (ii) Verifier-Side Security}, which guarantees existential unforgeability under a selective-ID, adaptive chosen message-and-ID attacks ($\eufsidcmia$, Definition~\ref{def:Experiment1}); and {\em (iii) Non-Repudiation}, which prevents signers from falsely denying valid signatures. These properties are quantified by distinct security parameters: $\lambda_1$ for signer-side fail-stop security, $\lambda_2$ for non-repudiation, and $\kappa$ for verifier-side unforgeability~\cite{susilo1999fail, safavi2000threshold, chen2021fail}.

\begin{definition} \label{def:Experiment2}
A $\hitsec$ provides $\lambda_1$-bit signer-side fail-stop security if, for any quantum-capable adversary \A controlling fewer than $t$-out-of-$n$ signers, the following holds: \vspace{-1mm}
\begin{equation*}
\scalebox{0.85}{$
\Pr\left[
\begin{aligned}
   &1 \as \hitsec.\mverify(m^\ast, \vec{\id}_k, \vec{Q}_{\id_k}, \sigma^\ast_k)~\land \\
   &0 \as \hitsec.\pof(\{\hat{e}^\ast_{i,\mathrm{j}}\}_{i=1}^{t}, \{\hat{d}^\ast_{i,\mathrm{j}}\}_{i=1}^{t}, m^\ast, \sigma^\ast_k, hist)
\end{aligned}
\right] \leq \mathit{negl}(\lambda_1)
$}
\end{equation*}\vspace{-1mm}
where the forged signature passes the verification, and $\hitsec.\pof$ is the proof generated by at least one honest signer. This bound holds as long as \A has not queried the signing oracle on $m^\ast$ nor obtained the commitment values ($\{{\hat{e}^\ast_{i,\mathrm{j}}\}_{i=1}^{t}, \{\hat{d}^\ast_{i,\mathrm{j}}\}_{i=1}^{t}}$) from uncorrupted signers.
\end{definition}


\begin{definition} \label{def:Experiment3}
A $\hitsec$ scheme provides $\lambda_2$-bit non-repudiation FS security if, for any quantum-capable $\mathcal{A}$ controlling one-out-of-$t$ participating signers, the following holds: \vspace{-1mm}
\begin{equation*}
\scalebox{0.85}{$
\Pr\left[
\begin{aligned}
   &1 \as \hitsec.\mverify(m, \vec{\id}_k, \vec{Q}_{\id_k}, \sigma_k)~\land \\
   &\pi^\ast \as \hitsec.\pof(\{\hat{e}^\ast_{i,\mathrm{j}}\}_{i=1}^{t}, \{\hat{d}^\ast_{i,\mathrm{j}}\}_{i=1}^{t}, m, \sigma^\ast_k, hist)
\end{aligned}
\right] \leq \mathit{negl}(\lambda_2)
$}
\end{equation*} 
where $\sigma_k$ is generated according to the signing protocol, and $\pi^\ast$ is the forgery proof which is not equal to "$\mathit{Not~A~Forgery}$". This bound holds as long as \A has not queried the commitment values ($\{\hat{e}^\ast_{i,\mathrm{j}}\}_{i=1}^{t}, \{\hat{d}^\ast_{i,\mathrm{j}}\}_{i=1}^{t}$) of any uncorrupted signers. 
\end{definition}

\vspace{-3mm}
\section{Feasibility Analysis of NIST-PQC for Initial Bootstrapping in 5G Cellular Network}
\label{sec:feasibility}



\subsection{Challenges in Deploying NIST-PQC Signatures} 
For any BS authentication mechanism, it is critical that the BS signs the SIB1 message and, ideally, embeds the signature within the same packet to avoid additional overhead~\cite{hussain2019insecure}. 
However, including the signature in SIB1 is challenging due to strict size constraints. 
According to the 3GPP RRC specification~\cite{RRCSpec}, the maximum allowable size for the SIB1 message is $2976$ bits or $372$ bytes (section 5.2.1). However, for initial communication, not all the bytes of the SIB1 are used. For instance, a minimally configured SIB1 typically occupies $80$–$100$ bytes. 
To assess real-world usage and configuration, we take measurements for SIB1 messages from two major U.S. vendors across multiple BSs--covering $8$ different physical cells. The observed SIB1s had a minimum size of $108$ bytes and a maximum of $120$ bytes. 
These findings suggest that, under current commercial vendor deployments, compact signatures can be piggybacked without requiring newly introduced messages.  
In contrast, larger signatures would require further modifications to network operations (see Section~\ref{subsec:frag_hurdle}). 

\looseness-1 Thus, to avoid fragmentation, the signature must fit within the current packet size alongside its standard fields. Unfortunately, NIST-selected signatures produce signature sizes that exceed this limit and are incompatible with the $372$-byte maximum size of SIB1. For example, the $\textit{FN\mbox{-}DSA}$~\cite{soni2021falcon} yields a $1280$-byte signature and a $1793$-byte public key at NIST security level 5, and even at level 1, it requires $666$ bytes for the signature and $897$ bytes for the public key. Similarly, $\textit{SLH\mbox{-}DSA}$ \cite{cooper2024stateless} produces a signature size of nearly $17$~KB at level 3, making it highly impractical for 5G UE-to-BS communication. The $\textit{ML\mbox{-}DSA}$ signature (e.g., $\textit{Dilithium2}$~\cite{dang2024module}) produces a $2420$-byte signature and a $1312$-byte public key. These sizes far exceed the SIB1 limit, and the only way to transmit such signatures would be to fragment both the signature and public key across multiple SIB1 packets.

\subsection{Fragmentation Constraints}
\label{subsec:frag_hurdle}
\looseness-1
Specifically, to transmit the additional $3732$ bytes required by $\textit{Dilithium2}$ ($2420$-byte signature and $1312$-byte public key), and assuming $290$ bytes of free space per SIB1, the BS would need to send $13$ separate SIB1 packets, each containing a fragment of the signature-key pair. The UE must then extract these fragments and reconstruct the full signature and key for verification. Notably, this overhead accounts for only a single level in the certificate chain; transmitting longer chains would exacerbate the problem. The resulting communication and processing burden on both ends substantially increases the setup-to-authentication time, making this approach impractical for 5G environments.


\looseness-1
We have already established the communication overhead of broadcasting large keys or certificates by fragmenting them across multiple SIB1 packets. However, the latency implications further exacerbate this challenge. According to the RRC specification, ``\textit{The SIB1 is transmitted on the DL-SCH with a periodicity of $160$~ms and variable transmission repetition periodicity within 160~m. The default transmission repetition periodicity of SIB1 is $20$~ms but the actual transmission repetition periodicity is up to network implementation.}" (section 5.2.1~\cite{RRCSpec}).   
This implies a default delay of $20$~ms between consecutive SIB1 packet transmissions and maximum delay of $160$~ms, even when broadcasting identical network parameters with different signature fragments. Consequently, transmitting these fragments introduces a baseline latency ranging from $20 \times 12 = 240$~ms upto $160 \times 12 = 1920$~ms, not including the time for packet generation at the BS and reconstruction at the UE.    
Moreover, the delivery of SIB1 can be unreliable over the DL-SCH channel. Therefore, each packet must be assigned a sequence number for identification.   
Now, in this unreliable scenario, even with the newly introduced sequence number, we need to consider the situation when the UE receives out-of-order packets and must wait for all $13$ packets on its end before reconstructing the key or signature. In the worst case, the UE can receive the $2$nd packet first, followed by the $3$rd, and so on; until it receives the $1$st packet, followed by the $2$nd; up to the $13$th ($p_2, p_3, \dots, p_{13}, p_1, p_2, \dots p_{12}, p_{13}$). Thus, under a uniformly random packet delivery model, the expected number of packets needed for successful key/signature delivery rises to $13 + \lfloor 13/2 \rfloor = 19$, resulting in an overall transmission delay ranging from $20 \times 18 = 360$~ms upto $160 \times 18 = 2880$~ms in the DL-SCH, rendering this approach unsuitable for time-sensitive 5G authentication. 

\looseness-1 
To handle the fragmentation scheme with a potential out-of-order packet arrival scenario, one needs to consider a sliding window process to keep track of the valid in-sequence SIB1s. 
The sliding window process can sort the incoming packets according to their sequence numbers. When all the required packets (in our ML-DSA single-chain example, there will be $13$ packets in the best case and $25$ packets in the worst case) are received, the process merges the extracted fragments to reconstruct the key/signature. 
In addition to the above challenges, transmitting multiple SIB1 packets adds significant complexity to lower protocol layers. For example, in the srsRAN open-source 5G stack, even when multiple SIB1 messages are generated, the MAC layer permits only a single packet for scheduling at a time. Moreover, the protocol expects each BS to broadcast one SIB1 message, after which the UE initiates connection procedures. 
These limitations further underscore the impracticality of existing PQC solutions in the context of 5G BS authentication.


\subsection{Comparative Analysis on 5G Testbed}
\noindent \textbf{Over-the-air Testbed Setup.} To understand the applicability of various cryptographic solutions for 5G UE authentication, we set up a Software Defined Radio (SDR)-based testbed and run the algorithms for over-the-air communication. We use srsRAN and open5GS open-source implementations for this purpose. More specifically, we run srsUE on a USRP $B210$ and srsgNB on another USRP $B210$--both connected to the same computer through USB $3.0$ ports. We also use a Leo Bodnar GPSDO to ensure a seamless $10~MHz$ external clock for the USRP devices. The attached srsgNB to the open5GS core gets connected to the srsUE through over-the-air communication. Since the baseband of commercial UEs cannot be modified, we adopt a best-effort approach and conduct our analysis using a widely used srsRAN-Open5GS testbed. Our setup is consistent with existing works on 4G/5G bootstrapping~\cite{ross2024fixing, hussain2019insecure}. Fig.~\ref{fig:5gtestbed} shows our complete testbed setup.

\begin{figure}
	\centering
	\includegraphics[scale=0.04]{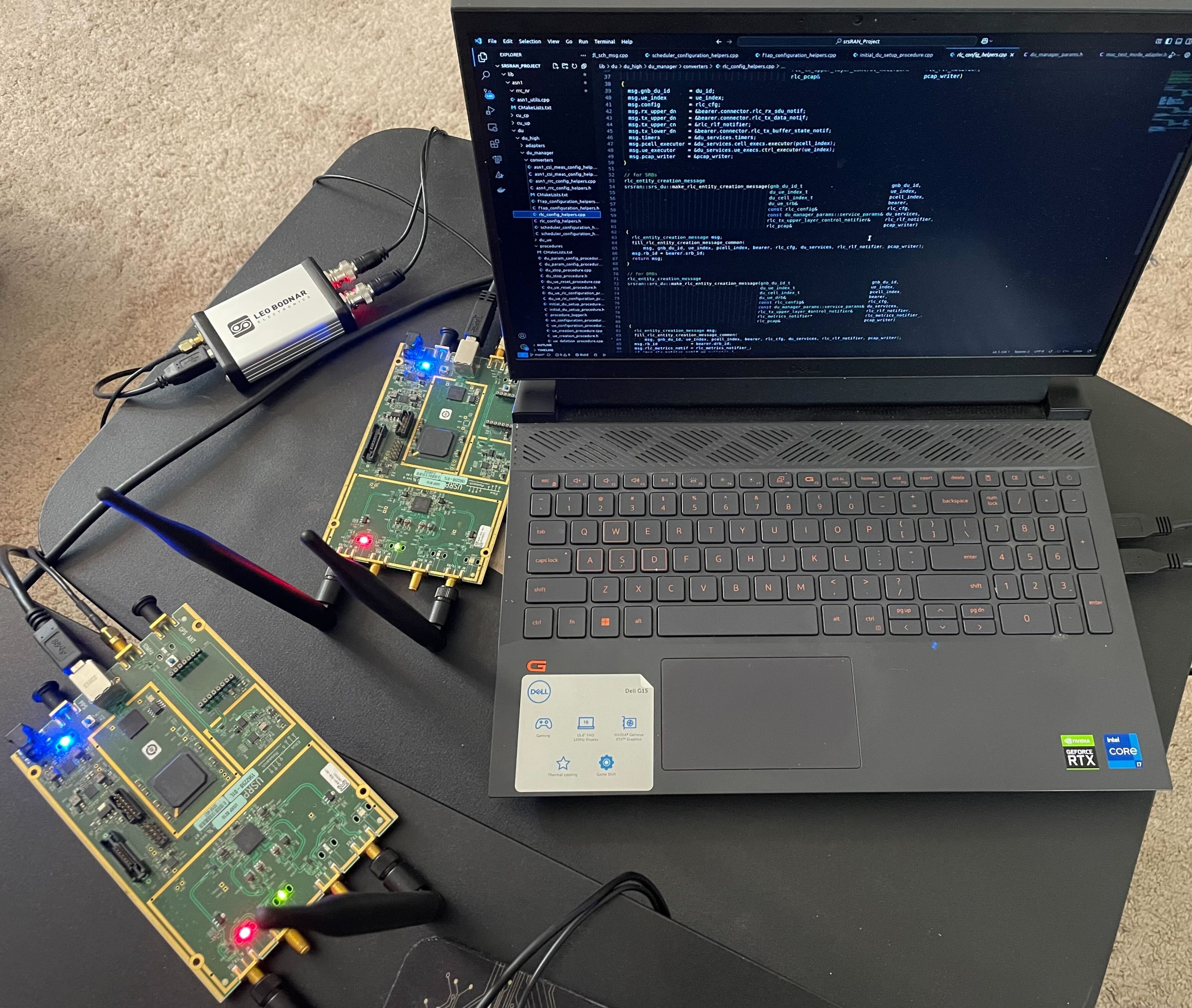}
	\caption{\small Testbed setup for 5G end-to-end communication. The gNB (USRP: bottom-left) and UE (USRP: top-right) are connected to a laptop and GPSDO.} \vspace{-5mm}
	\label{fig:5gtestbed}
\end{figure}

We evaluate existing schemes alongside our solution, $\borg$, in an over-the-air setup, measuring both cryptographic and network-induced delays. For $\textit{ML\mbox{-}DSA}$ with two certificate chains where fragmentation becomes essential, authentication requires $33$ additional packets, introducing roughly $12~\textit{KB}$ of overhead. In contrast, $\borg$ requires no additional packets, making it highly compatible with current 5G protocol constraints. Also, $\textit{ML\mbox{-}DSA}$ incurs a latency of approximately ranging from $662.47$~ms upto $5282.47$~ms, which is nearly ${221\sim1767\times}$ higher than that of $\borg$ (see TABLE~\ref{tab:PerformanceComparison3}), rendering it impractical. While schemes like $\textit{EC\mbox{-}Schnorr}$ avoid fragmentation, they lack PQ forgery detection and compromise resilience. This puts $\borg$ in a suitable spot of compatibility both from a PQ and 5G network perspective. Further performance details are provided in Section~\ref{sec:PerformanceEvaluation}. \vspace{-2mm}

\section{The Proposed Future-Proof Solution} 
\label{sec:solution}
 \noindent We begin with an overview of $\borg$, followed by algorithm descriptions and protocol instantiation for 5G networks.

\subsection{The Proposed $\borg$ Scheme} \label{subsec:scheme}
Given the infeasibility of current NIST-PQC standards, $\borg$ focuses on conventional secure techniques enhanced with key features: threshold signing for distributed trust, IBS to lift certificate burdens, and fail-stop mechanisms with PQ threshold audit logging for post-mortem PQ forgery detection. We present $\borg$ in Algorithms \ref{Alg:HITFSS1}-\ref{Alg:HITFSS4}.

\begin{algorithm}[ht!]
	\small
	\caption{$\borg$ \textit{(Setup and Key Extraction)}}\label{Alg:HITFSS1}
	\hspace{5pt}
\begin{algorithmic}[1]
\Statex \vspace{-1mm}\hspace{-5mm}$\underline{(\sk_{\id_0}, \pk_{\id_0}, params) \as \borg.\setup (1^\kappa)}$: CKG runs this algorithm once to set up the system. 

\State $\alpha_0 \asrand \mathbb{Z}_q$, ${\sk_{\id_0} \as H_1(\alpha_0)}$, and ${\pk_{\id_0} \as g^{\sk_{\id_0}} \mod p}$

\State \textbf{return}~${\sk_{\id_0}}$, ${\pk_{\id_0}}$, and ${params \as \{p, q, H_1, H_2\}}$\vspace{+1mm}
\end{algorithmic}
\begin{algorithmic}[1]
\Statex \hspace{-5mm}$\underline{(\{\sk_{\id_{k,i}}\}_{i=1}^{n}, \{\pk_{\id_{k,i}}\}_{i=1}^{n}, \vec{Q}_{\id_k}) \as \borg.\keyextract(\id_k}$, $\underline{\vec{Q}_{\id_{k-1}},\sk_{\id_{k-1}})}$: This algorithm is run by the user at level $(k-1)$ to generate key pairs for users at level $k$. 

\State ${\alpha_k \asrand \zq}$, ${r_k \as H_1(\alpha_k)}$, and ${Q_{\id_k} \as g^{r_k} \mod p}$

\State ${\vec{Q}_{\id_k} \as (\vec{Q}_{\id_{k-1}}, Q_{\id_k})}$ and ${h_{\id_k} \as H_1(\id_k || \vec{Q}_{\id_k})}$

\State ${\sk_{\id_k} \as sk_{\id_{k-1}} \cdot h_{\id_k} + r_k \mod q}$

\State ${f(x) = \sk_{\id_k} + \sum\limits_{i=1}^{t-1}a_i \cdot x^i \mod q}$, where ${a_i\asrand\zq}$, $i\in\{1,...,{t-1}\}$
\State \textbf{for}~{$i = 1, 2, \dots, n$}~\textbf{do}
\State ~~~~${\sk_{\id_{k,i}} \as f(i)}$ and ${\pk_{\id_{k,i}} \as g^{\sk_{\id_{k,i}}} \mod p}$ 
\State ~~~~\textbf{return}~$(\{\sk_{\id_{k,i}}\}_{i=1}^{n}, \{\pk_{\id_{k,i}}\}_{i=1}^{n}, \vec{Q}_{\id_k})$
\end{algorithmic}
\end{algorithm} 
\setlength{\textfloatsep}{0pt}

$\borg.\setup$ (Algorithm \ref{Alg:HITFSS1}) initializes the system by generating the master secret and public keys $(\sk_{\id_0}, \pk_{\id_0})$, and publishing the public parameters $params$.  
${\borg.\keyextract}$ (Algorithm \ref{Alg:HITFSS1}) enables hierarch-ical key extraction, allowing each level to derive key pairs for the subsequent level. Specifically, users at level ${k-1}$ utilize their secret key $\sk_{\id_{k-1}}$, the public key vector ${\vec{Q}_{\id_{k-1}} = \{\pk_{\id_{0}}, Q_{\id_1},\dots,}$ ${Q_{\id_{k-1}}\}}$, and the group identity $\id_k$ of the lower level to compute the group verification key $Q_{\id_k}$ and derive individual key pairs $(\sk_{\id_{k,i}}, \pk_{\id_{k,i}})$ for each user ${i = 1, \dots, n}$ at level $k$. Secret keys are then securely distributed to each user. The derived secret key follows a Schnorr signature structure and can be reconstructed via Lagrange interpolation \cite{shamir1979share} from any $t$-out-of-$n$ users at level $k$: $\sk_{\id_{k}} = \sum_{i=1}^{t} \lambda_i\cdot\sk_{\id_{k,i}}~mod~q$.




\begin{algorithm}[ht!]
	\small
	\caption{$\borg$ \textit{(Preprocessing)}}\label{Alg:HITFSS2}
	\hspace{5pt}
\begin{algorithmic}[1]
\Statex \vspace{-1mm}\hspace{-5mm}$\underline{\mathcal{L}_{i} \as \borg.\preprocess(J)}$: $n$ users at level $k$ execute this algorithm to generate commitment values, enabling them to sign up to $J$ messages.
\State \textbf{for}~{$i = 1, 2, \dots, n$}~\textbf{do}
\State ~~~~\textbf{for}~{$j = 1, 2, \dots, J$}~\textbf{do}
\State ~~~~~~~~~${(\hat{e}_{i,j}, \hat{d}_{i,j}) \asrand \zq \times \zq}$ 
\State ~~~~~~~~~${e_{i,j} \as H_1(\hat{e}_{i,j} || j || \id_{k,i})}$, ${d_{i,j} \as H_1(\hat{d}_{i,j} || j || \id_{k,i})}$
\State ~~~~~~~~~${E_{i,j} \as g^{e_{i,j}} \mod p}$ and ${D_{i,j} \as g^{d_{i,j}} \mod p}$
\State ~~~~Send ${(\{E_{i,j}\}_{j=1}^{J}, \{D_{i,j}\}_{j=1}^{J})}$ to the other ${n-1}$ users.
\State \textbf{for}~{$i = 1, 2, \dots, n$}~\textbf{do}
\State ~~~~\textbf{if}~{$(\{E_{i,j}\}_{j=1}^{J}, \{D_{i,j}\}_{j=1}^{J}) \in \mathbb{G}$}~\textbf{then}
\State ~~~~~~~~~${\mathcal{L}_{i,j} \as (E_{i,j}, D_{i,j})}$
\State ~~~~\textbf{return}~${\mathcal{L}_{i} \as (i, \{E_{i,j}\}_{j=1}^{J}, \{D_{i,j}\}_{j=1}^{J})}$
\EndFor
\end{algorithmic}
\end{algorithm} 
\setlength{\textfloatsep}{0pt}

\looseness-1 
Akin to Schnorr-style threshold signatures (i.e., \cite{komlo2021frost}), the signing follows a two-round protocol with three phases: (i) preprocessing: to generate a shared list of random commitments values, (ii) threshold signing: allows participants to compute their signature shares, and (iii) aggregation: to combine these shares into a group signature. 
The $\borg.\preprocess$ phase (Algorithm \ref{Alg:HITFSS2}) is jointly executed by all $n$ users at level $k$ to generate random commitment values required for signing up to $J$ messages. Each user ${i \in \{1, \dots, n\}}$ computes commitment pairs $(E_{i,j}, D_{i,j})$ for $j=1$ to $J$, and shares them with the others. After verifying their validity, users append the values to the commitment list $\mathcal{L}_{i,j}$, resulting in a consistent finalized list $\mathcal{L}_i$ across all participants. In practice, $\mathcal{L}_i$ may be published at a predefined location (e.g., a public bulletin) accessible to all users. 

In $\borg.\sign$ (Algorithm \ref{Alg:HITFSS3}), each signer uses its secret key share $\sk_{\id_{k,i}}$ and the commitment list $\mathcal{L}_{i}$ to compute a signature share for message $m$ at index ${j \in \{1,\dots, J\}}$, and submits it to the other $\beta$ signing participants. The signer set size $\beta$ is predetermined (${t \leq \beta \leq n}$) and must meet the threshold $t$ to produce a valid group signature. Upon receiving ${\beta-1}$ shares, each signer verifies them individually (aborting on failure) and aggregates the valid shares into the group signature $\sigma_{k,j} = (R_j, z_j)$.  
%
The $\borg.\mverify$ (Algorithm \ref{Alg:HITFSS3}) follows the standard Schnorr signature verification process. Using the vector of public keys ${\vec{Q}}_{\id_k}$ and identities $\vec{\id_k}$, the verifier validates the signature $\sigma_{k,j}$ on the message $m_j$. The correctness of the verification algorithm resembles Schnorr-based signature verification~\cite{singla2021look} and is given in Section~\ref{sec:verificationcorrectness}.

\begin{algorithm}[ht!]
	\small
	\caption{$\borg$ \textit{(Message Signing and Verification)}}\label{Alg:HITFSS3}
	\hspace{5pt}
\begin{algorithmic}[1]
\Statex \vspace{-1mm}\hspace{-5mm}$\underline{\sigma_{k,j} \as \borg.\sign(m_j, \mathcal{L}_{i}, \{\sk_{k,i}\}_{i=1}^{\beta})}$: At level $k$, $\beta$ participants ($\beta \in [t,n]$) execute this algorithm: 
\State \textbf{for}~{$i = 1, 2, \dots, \beta$}~\textbf{do}
\State ~~~~$\rho_{i,j} \as H_1(i || m_j || \{\mathcal{L}_{i,j}\}_{i=1}^{\beta})$ 
\State ~~~~$R_{i,j} \as D_{i,j} \cdot (E_{i,j})^{\rho_{i,j}} \mod p$
\State ~~~~${R_j \as \prod\limits_{i=1}^{\beta}R_{i,j} \mod p}$ and $h_j \as H_2(R_j || Q_{\id_k} || m_j)$
\State ~~~~$z_{i,j} \as d_{i,j} + e_{i,j} \cdot \rho_{i,j} + \lambda_i \cdot \sk_{\id_{k,i}} \cdot h_j \mod q$ 
\State Send $\{z_{i,j}\}_{i=1}^{\beta}$ to ${\beta-1}$ participants.
\Statex Each participant $i$ performs:
\State \textbf{for}~{$i = 1, 2, \dots, \beta$}~\textbf{do}
\State ~~~~${\rho_{i,j} \as H_1(i || m_j || \{\mathcal{L}_{i,j}\}_{i=1}^{\beta})}$ 
\State ~~~~${h_j \as H_2(R_j || Q_{\id_k} || m_j)}$
\State ~~~~\textbf{if}~{${g^{z_{i,j}} \neq R_{i,j} \cdot \pk_i^{\lambda_i \cdot h_j}\mod p}$}~\textbf{then} 
\State ~~~~~~~~\textbf{return} $\perp$. 
\State ~~~~\textbf{else}~${z_j \as \sum\limits_{i=1}^{\beta}z_{i,j}}\mod q$ and ${R_j \as \prod\limits_{i=1}^{\beta}R_{i,j}}\mod p$ 
\State ~~~~~~~~~\textbf{return}~${\sigma_{k,j} \as (R_j, z_j)}$
\end{algorithmic}
%
%
%
%
\begin{algorithmic}[1]
\Statex \hspace{-5mm} $\underline{\{0, 1\} \as \borg.\mverify(m_j, {\vec{\id}}_k, {\vec{Q}}_{\id_k}, \sigma_{k,j})}$: This algorithm is executable by any user within the network.
\State $h_{\id_\ell} \as H_1(\id_\ell || \vec{Q}_{\id_\ell})$ for $\ell = 1, 2, \dots, k$
\State $Q \as \prod\limits_{\ell=1}^{k-1} (Q_{\id_\ell})^{\prod\limits_{\omega=\ell+1}^{k}h_{\id_\omega}}$
\State $h_j \as H_2(R_j || Q_{\id_k} || m_j)$ 
\State \textbf{if}~{$g^{z_j} \stackrel{\mbox{?}}{=} R_j \cdot (Q \cdot Q_{\id_k} \cdot (\pk_{\id_0})^{\prod\limits_{\ell=1}^{k}h_{\id_\ell}})^{h_j} \mod p$}~\textbf{then}, \textbf{return}~$1$ 
\end{algorithmic}
\end{algorithm} 
\setlength{\textfloatsep}{0pt}

\begin{algorithm}[ht!]
	\small
	\caption{$\borg$ \textit{(Forgery Detection and Verification)}}\label{Alg:HITFSS4}
	\hspace{5pt}
\begin{algorithmic}[1] 
\Statex \hspace{-5mm} $\underline{\pi \as \borg.\pof(\{\hat{e}_{i,\mathrm{j}}\}_{i=1}^{\beta}, \{\hat{d}_{i,\mathrm{j}}\}_{i=1}^{\beta}, m_\mathrm{j}, \sigma_k', hist)}$: This algorithm is run by the $\beta$ signing participants at level $k$ to prove the forgery of a signature $\sigma_k'$ to entities at level ${k-1}$: 
\State $\mathrm{j} \as hist$ and $(R_\mathrm{j}', z_\mathrm{j}') \as \sigma_k'$
\State \textbf{for}~{$i = 1, 2, \dots, \beta$}~\textbf{do}
\State ~~~~Reveal $(\hat{e}_{i,\mathrm{j}}, \hat{d}_{i,\mathrm{j}})$ to other ${\beta-1}$ participants 
\Statex Each participant $i$ performs:
\State \textbf{for}~{$i = 1, 2, \dots, \beta$}~\textbf{do}
\State ~~~~$e_{i,\mathrm{j}} \as H_1(\hat{e}_{i,\mathrm{j}} || \mathrm{j} || m_\mathrm{j})$ and $d_{i,\mathrm{j}} \as H_1(\hat{d}_{i,\mathrm{j}} || \mathrm{j} || m_\mathrm{j})$ 
\State ~~~~$E_{i,\mathrm{j}} \as g^{e_{i,\mathrm{j}}} \mod p$ and $D_{i,\mathrm{j}} \as g^{d_{i,\mathrm{j}}} \mod p$
\State ~~~~$\rho_{i,\mathrm{j}} \as H_1(i || m_\mathrm{j} || \{\mathcal{L}_{i,\mathrm{j}}\}_{i=1}^{\beta})$
\State $R_\mathrm{j} \as \prod\limits_{i=1}^{\beta} D_{i,\mathrm{j}} \cdot (E_{i,\mathrm{j}})^{\rho_{i,\mathrm{j}}} \mod p$
\State \textbf{if}~{$R_\mathrm{j}' = R_\mathrm{j}$}~\textbf{then}
\State ~~~~\textbf{return}~$\pi \as$ ``\textit{Not A Forgery}"
\State \textbf{if}~{$R_\mathrm{j}' \neq R_\mathrm{j}$}~\textbf{then}
\State ~~~~\textbf{return}~$\pi \as (\mathrm{j}, \{\hat{e}_{i,\mathrm{j}}\}_{i=1}^{\beta}, \{\hat{d}_{i,\mathrm{j}}\}_{i=1}^{\beta})$
\end{algorithmic}
%
%
%
%
\begin{algorithmic}[1]
\Statex \hspace{-5mm}$\underline{\{0, 1\} \as \borg.\fverify(\alpha_{k}, \sk_{\id_{k-1}}, {\vec{Q}}_{\id_k}, m, \sigma_k', \pi)}$: This algorithm is run by the level $k-1$ to verify the proof of forgery. 
\State \textbf{if}~{$\pi' = ``\mbox{\textit{Not A Forgery}}"$}~\textbf{then}
\State ~~~~\textbf{return}~$\perp$
\State \textbf{if}~{$\pi' = (\mathrm{j}, \{\hat{e}_{i,\mathrm{j}}\}_{i=1}^{\beta}, \{\hat{d}_{i,\mathrm{j}}\}_{i=1}^{\beta})$}~\textbf{then}
\State ~~~~$r_{k} \as H_1(\alpha_{k})$ and $Q_{\id_{k}} \as g^{r_{k}} \mod p$
\State ~~~~$h_{\id_k} \as H_1(\id_k || {\vec{Q}}_{\id_k})$
\State ~~~~\textbf{for}~{$i = 1, 2, \dots, \beta$}~\textbf{do}
\State ~~~~~~~~~\textbf{if}~{$\prod\limits_{i=1}^{\beta}\pk_{\id_{k,i}}^{\lambda_i} \neq (g^{h_{\id_{k}} \cdot \sk_{\id_{k-1}}}) \cdot Q_{\id_{k}}$}~\textbf{then}
\State ~~~~~~~~~~~~~~~~\textbf{return}~$\perp$
\State ~~~~~~~~~\textbf{else} $(R_\mathrm{j}', z_\mathrm{j}') \as \sigma_k'$
\State ~~~~\textbf{for}~{$i = 1, 2, \dots, \beta$}~\textbf{do}
\State ~~~~~~~~~$e_{i,\mathrm{j}} \as H_1(\hat{e}_{i,\mathrm{j}} || \mathrm{j} || m_\mathrm{j})$ and $d_{i,\mathrm{j}} \as H_1(\hat{d}_{i,\mathrm{j}} || \mathrm{j} || m_\mathrm{j})$ 
\State ~~~~~~~~~$E_{i,\mathrm{j}} \as g^{e_{i,\mathrm{j}}} \mod p$ and $D_{i,\mathrm{j}} \as g^{d_{i,\mathrm{j}}} \mod p$
\State ~~~~~~~~~$\rho_{i,\mathrm{j}} \as H_1(i || m_\mathrm{j} || \{\mathcal{L}_{i,\mathrm{j}}\}_{i=1}^{\beta})$
\State ~~~~$R_\mathrm{j} \as \prod\limits_{i=1}^{\beta} D_{i,\mathrm{j}} \cdot (E_{i,\mathrm{j}})^{\rho_{i,\mathrm{j}}} \mod p$
\State ~~~~\textbf{if}~{$R_\mathrm{j}' = R_\mathrm{j}$}~\textbf{then}, \textbf{return}~$0$ 
\State ~~~~\textbf{else}, \textbf{return}~$1$
\end{algorithmic}
\end{algorithm}
\setlength{\textfloatsep}{0pt}

Multiple valid signatures may satisfy the verification algorithm; however, even if a capable adversary obtains such signatures, they cannot distinguish those genuinely generated by authorized signers \cite{susilo1999fail}. Leveraging this principle and following fail-stop mechanisms (e.g., \cite{susilo1999fail, yaksetig2024extremely}), signers at level $k$ will claim forgery by invoking the forgery proof algorithm $\borg.\pof$, while higher-level authorities at level ${k-1}$ run $\borg.\fverify$ to validate the claim.

$\borg.\pof$ (Algorithm \ref{Alg:HITFSS4}) is invoked by signers at level $k$ upon detecting a suspected forgery of message $m_j$ (e.g., via signature audit logs). Using the signature history ($hist$) to identify the message index $\mathrm{j}$, the $\beta$ participating signers reveal their secret nonces $(\{\hat{e}_{i,\mathrm{j}}\}_{i=1}^{\beta}, \{\hat{d}_{i,\mathrm{j}}\}_{i=1}^{\beta})$ originally generated during preprocessing. Each signer reconstructs the signature component $R_\mathrm{j}$ as in the signing process and compares it with the corresponding component in the suspected signature $\sigma_k'$. If they match, the signer outputs $\pi$ as "\textit{Not a Forgery}"; otherwise, it outputs the proof $\pi = (\mathrm{j}, \{\hat{e}_{i,\mathrm{j}}\}_{i=1}^{\beta}, \{\hat{d}_{i,\mathrm{j}}\}_{i=1}^{\beta})$, which is submitted to the higher-level authority at level ${k-1}$.  
$\borg.\fverify$ (Algorithm \ref{Alg:HITFSS4}) is executed by the level ${k-1}$ authority. Using the secret value $\alpha_k$ chosen during key extraction, the verifier first checks the validity of the public and group verification keys. If this fails, the proof is rejected. Otherwise, it reconstructs $R_\mathrm{j}$ from the disclosed nonces in $\pi$ and compares it with $R_\mathrm{j}'$ from the suspected forged signature. A match results in rejection of the forgery claim; a mismatch confirms forgery, prompting system halt due to a detected security breach.

\subsection{Security Analysis} 
\label{subsec:securityanalysis}
Following Definition~\ref{def:Experiment1}, we prove $\mathit{EUF}\mbox{-}\mathit{sID}\mbox{-}\mathit{CMIA}$ security of $\borg$ under the hardness of the (Elliptic Curve)-DLP based on the generalized forking lemma in the random oracle model \cite{bellare2006multi, boldyreva2012secure}. Additionally, based on Definition~\ref{def:Experiment2}, we prove the signer-side fail-stop security of $\borg$ under the hardness of second preimage resistance of a cryptographically secure hash function.

\vspace{-1mm}
\begin{theorem} \label{the:verifier-side} 
\textit{If an adversary \A can $(q_{E}, q_P, q_S, q_{H_1,H_2})$-break $\borg$ in the random oracle model (Definition \ref{def:Experiment1}) with an advantage $\epsilon$ in time $\tau$ while having access to at most ${(t\mbox{--}1)}$-out-of-$n$ signing participants, where $q_E, q_P, q_S, q_{H_1}$, and $q_{H_2}$ denote queries to key extraction, preprocessing, signing, and hash functions $H_1$ and $H_2$, then an algorithm \C can be constructed to break the (EC)DLP in group $\mathbb{G}$.} 
\end{theorem}

\begin{theorem}\label{theorem:FSsecurity}
$\borg$ provides $\lambda_1$-bit signer-side fail-stop security against quantum-capable adversaries controlling up to ${(t-1)}$-out-of-$n$ signing participants, as formalized in Definition~\ref{def:Experiment2}, and $\lambda_2$-bit non-repudiation security against quantum adversaries with access to one-out-of$t$ signing participants, as captured in Definition~\ref{def:Experiment3}. Both guarantees rely on the hardness of breaking the second preimage resistance of a cryptographically secure hash function, while providing $\kappa$ bit verifier-side security via $\eufsidcmia$ in the random oracle model as defined in Definition~\ref{def:Experiment1}.\vspace{+1mm}
\end{theorem}\vspace{-2mm}

\noindent Full security proofs are presented in Appendix~\ref{subsec:securityproof}.

\subsection{Instantiation of $\borg$ for 5G Network} 
This section outlines the instantiation of the $\borg$ algorithm for 5G and beyond mobile networks, providing a high-level overview while detailing each step in Figure~\ref{fig:5gprotocol}.

\noindent \textbf{Solution Architecture.} The proposed architecture adopts a two-layer design comprising: {\em (1)} a newly introduced Core Key Generator (CKG) for key generation and system initialization within the 5G core; {\em (2)} the Access and Mobility Management Function (AMF); {\em (3)} a set of Base Stations (BSs); and {\em (4)} User Equipment (UE), representing end-user devices such as smartphones, laptops, and IoT nodes.

\begin{figure}
	\centering
	\includegraphics[scale=0.37]{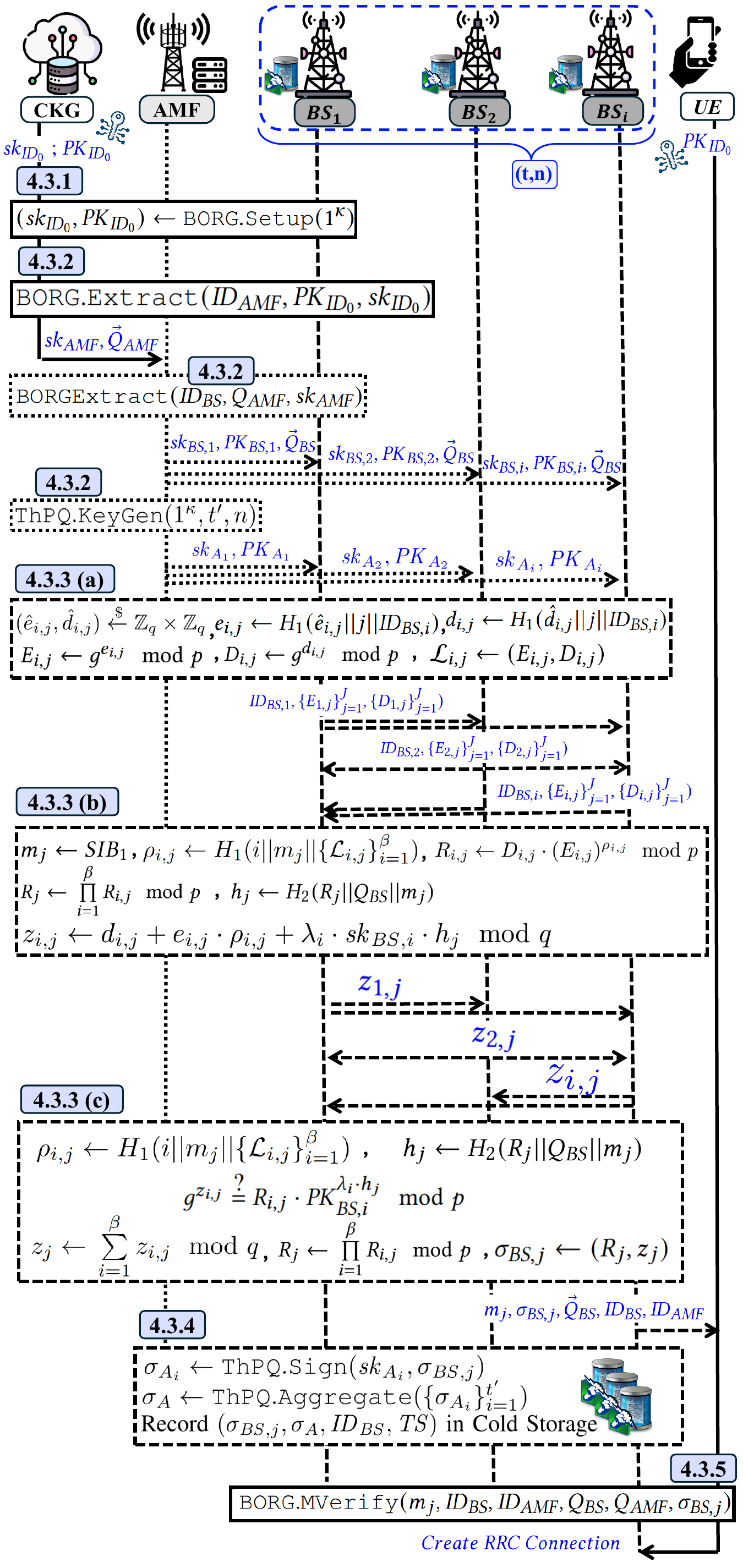}
	\caption{Our protocol for authenticating 5G cellular BSs}
	\label{fig:5gprotocol}
\end{figure}

\noindent \textbf{Our Authentication Protocol Overview:} The CKG manages key setup for the AMF, which in turn handles key extraction for all BSs. Key extraction occurs periodically. Each entity is identified by an $\id$ concatenated with a secret share expiry timestamp, used as input in key extraction. CKG key pairs, valid for years, are pre-installed in the UE’s USIM. BS keys default to a one-day validity, adjustable as needed. Due to physical exposure, BSs benefit from the $(t,n)$ threshold, requiring compromise of multiple entities, making frequent updates (e.g., hourly) both effective and practical. These periods are configurable by the 5G core.   

Key extraction is efficient for both AMF and CKG, scaling efficiently even at a global level. The derived secret keys are securely distributed from the AMF to the BSs using authenticated control channels. In practice, this can be achieved via the Xn Application Protocol (XnAP)—which supports secure inter-gNB signaling—or through out-of-band non-3GPP access mechanisms similar to those used for eSIM remote provisioning, ensuring mutual authentication, confidentiality, and integrity during key transfer~\cite{ETSIxnap}.   
Similarly, Base Station IDs can be distributed using the Protocol Extension Container Information Element (\textit{ProtocolExtensionContainer IE}) available in the XnAP messages. This IE is particularly kept to carry additional elements for communication, all while keeping backward compatibility.
The $(t,n)$ threshold ensures that at least $t$ BSs collaborate to generate a valid signature, with nonces preprocessed in batches to further reduce costs. 
All participating BSs also sign the resulting $SIB$ signature using their $\thpq$ secret shares and record the final audit signature in distributed cold storage~\cite{yavuz2019system, nouma2023practical, le2022efficient}. Since each BS holds the aggregated SIB signature and independently logs it using the $(t',n)$ threshold audit signature (where ${t\leq t'}$) via the $\thpq$ scheme, we eliminate any single point of failure. This approach offers a scalable solution for long-term archival in large-scale 5G deployments and enables fault-tolerant PQ threshold audit logging to support PM forgery detection. 

The UE verifies the SIB1 signature using only the pre-installed master key, with other public key data sent alongside the signature, ideal for resource-constrained devices. In the event of forgery, BSs disclose their nonces to generate a proof, which is sent to the AMF, enabling a system halt if needed. Precisely, the FS mechanism allows the core network to isolate compromised BSs and suspend authentication, preventing forged SIBs from spreading. Upon provable detection, the system halts affected operations, enabling swift, policy-enforced recovery with minimal disruption.






\subsubsection{System Initialization Phase}\label{subsubsec:1} 
In the two-layered architecture, the CKG handles the one-time system setup during initial mobile network deployment by executing $\borg.\setup$ as defined in Algorithm~\ref{Alg:HITFSS1}.

\subsubsection{Key Extraction Phase}\label{subsubsec:2}
\looseness-1 This phase initiates with the CKG running $\borg.\keyextract$ (Algorithm~\ref{Alg:HITFSS1}) to derive the AMF's key pair $(\sk_\amf, \vec{Q}_\amf)$ from the master secret $\sk_{\id_0}$, which is then securely transmitted to the AMF. The AMF then applies the same procedure to generate threshold-shared keys for the BSs. Under a $(t, n)$ configuration (e.g., 2-of-3), each $\bs_i$ receives a share $\sk_{\bs,i}$ and public key $\pk_{\bs,i}$, with all BSs sharing a group verification key $\vec{Q}_\bs$. Any $t$ out of $n$ can jointly produce a valid signature verifiable by $\vec{Q}_\bs$. Additionally, the AMF executes $\thpq.\keygen$ to produce the audit public key $\pk_A$ and secret key shares $(\sk_{A_1}, \dots, \sk_{A_n})$, enabling any $t' \geq t$ BSs to generate a valid $\thpq$ signature for secure audit logging. All keys are distributed over secure channels.

\subsubsection{Signing Broadcast Messages Phase} \label{subsubsec:3}
To sign SIB1, $\beta \in [t, n]$ BSs collaboratively generate a group signature through three phases: {\em (a) Preprocessing:} As defined in $\borg.\preprocess$ (Algorithm~\ref{Alg:HITFSS2}), this phase is precomputed for a given window (e.g., daily). Each $\bs_i$ uses its NRCell\_ID ($\id_{\bs,i}$) to generate random nonces and compute commitment pairs $(\{E_{i,j}\}_{j=1}^{J}, \{D{i,j}\}_{j=1}^{J})$ for $J$ messages. These are exchanged, verified, and appended to local commitment lists $\mathcal{L}_{i,j}$, yielding a consistent list $\mathcal{L}_i$ shared across all BSs, optionally published in a predefined location (e.g., a bulletin).   
{\em (b) Threshold Signing:} Following $\borg.\sign$ (Algorithm~\ref{Alg:HITFSS3}), each $\bs_i$ uses its key share and commitment list $\mathcal{L}_i$ to sign message ${m_j \as SIB_1}$, then sends the resulting signature share to the remaining ${\beta-1}$ signers. The number of signers $\beta$ is predetermined prior to signing.  
{\em (c) Aggregation:} Upon receiving the signature shares, each BS verifies the signature shares and aggregates the valid ones into the final group signature $\sigma_{\bs,j} = (R_j, z_j)$. 
The BS with the strongest signal strength broadcasts the $(m_j, \sigma_{\bs,j}, {\vec{Q}}_\bs, \id_\bs, \id_\amf)$.

\subsubsection{Audit Logging Phase}\label{subsubsec:5}  
\looseness-1
To strengthen forgery detection and independent of the 5G signing process, a subset of $t'$ participating BSs (where ${t\leq t'}$) collaboratively sign the SIB1 signature $\sigma_{\bs,j}$. Using $\thpq.\sign$, each $\bs_i$ signs $\sigma_{\bs,j}$ and sends the share to the other participants. Upon collecting $t'$ valid $\thpq$ shares, each BS runs $\thpq.\aggregate$ to produce the final threshold signature $\sigma_A$. The audit log entry $(\sigma_{\bs,j}, \sigma_A, \id_\bs, \ts)$ is then stored in distributed cold storage, ensuring fault-tolerant threshold auditability and and supporting proof-of-forgery verification against quantum adversaries.  In practice, the distributed cold storage is realized as a lightweight, append-only audit repository replicated across trusted RAN/Core entities (e.g., AMF, gNBs), where updates occur periodically via authenticated control channels (e.g., XnAP, TLS). This design imposes negligible signaling overhead while maintaining tamper-evident, verifiable records for post-mortem analysis.

\subsubsection{Signature Verification Phase}
\label{subsubsec:4}
The signature verification process follows $\borg.\mverify$ (Algorithm~\ref{Alg:HITFSS3}). Given a signature on the SIB1 message, the group verification key $Q_\bs$, the AMF's public key $Q_\amf$, and the CKG’s master public key $\pk_{\id_0}$ (pre-installed on the user’s device), the UE verifies the broadcast message to authenticate the BSs before initiating an RRC connection.

\subsubsection{Forgery Detection and Verification Phase} \label{subsubsec:7}  
Each BS has access to the authenticated distributed cold storage that records all $SIB$ messages and aggregated signatures. Upon suspecting forgery, a BS initiates the PoF protocol with the other $t$ participating signers, as illustrated in Figure~\ref{fig:forgery}. For a suspected forged signature $\sigma'_{\bs}$ on a SIB1 message, the $\beta$ participating signers identify the message index $\mathrm{j}$ from the cold storage log and reveal their corresponding preprocessing nonces $(\hat{e}_{i,\mathrm{j}}, \hat{d}_{i,\mathrm{j}})$ for $i = 1, \dots, \beta$. They invoke $\borg.\pof$ (Algorithm~\ref{Alg:HITFSS4}) to reconstruct $R_j$ and evaluate the validity of $\sigma'_{\bs}$.  Forgery verification is performed by the AMF using $\borg.\fverify$. Upon receiving the proof $\pi$ from the signing BSs, the AMF verifies the identities and public keys of the involved BSs, reconstructs $R_j$ from $\pi$, and compares it with $R_j'$ in $\sigma'_{\bs}$. A mismatch indicates a breach in the security of the authentication system and provably confirms that the underlying security assumption (i.e., ECDLP) has been broken, prompting a scoped, policy-gated response rather than an indiscriminate shutdown: the AMF isolates and re-keys the implicated $(t,n)$ signing group, while UEs fall back to scanning for alternative authenticated BSs per the configurable behavior of Section~\ref{subsec:authfailact}, preserving availability.

\noindent\textbf{Resistance to halt-induced DoS.} A halt is triggered only by a \textbf{provable} forgery, i.e., a signature that satisfies \borg . \mverify~yet does not match the legitimate signers' committed nonces $R_j$ (Algorithm~\ref{Alg:HITFSS4}). An adversary injecting arbitrary or malformed signatures fails verification at the UE; such messages are discarded during normal SIB1 processing (Section~\ref{subsec:authfailact}) and never reach the \pof / \fverify~path, so they cannot induce a halt. Producing a signature that \textit{passes} \mverify~while differing from the legitimate one requires solving the (EC)DLP (Theorem~\ref{the:verifier-side}) or quantum capability; for a computationally bounded adversary, this is infeasible, leaving no low-cost path to a detection event. The only way to force a halt is therefore to perform precisely the cryptographic break that the fail-stop property is designed to expose. Furthermore, an insider controlling fewer than $t$ signers cannot fabricate a forgery claim against an honestly generated signature: by the $\lambda_2$-bit non-repudiation guarantee (Definition~\ref{def:Experiment3}, Theorem~\ref{theorem:FSsecurity}), a \fverify-accepting proof $\pi^*$ for a legitimate signature would require a second preimage of $H_1$, and the AMF independently re-validates each claim via \borg .\fverify~using $\alpha_k$ before acting. Thus, neither external injection nor insider misbehavior provides a controllable DoS trigger.

\begin{figure}
	\centering
	\includegraphics[scale=0.38]{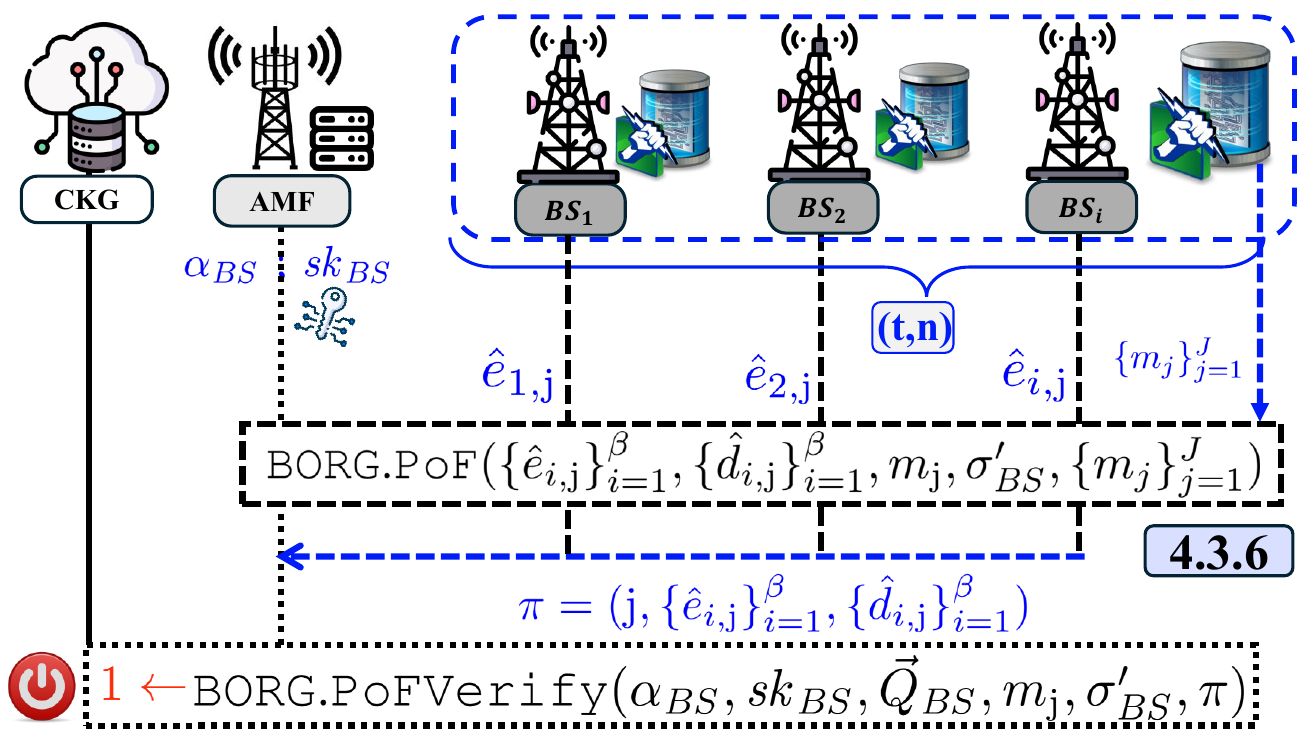}
	\caption{Instantiation of Forgery Proof and Verification}
	\label{fig:forgery}
\end{figure}


\subsubsection{Authentication failure action}\label{subsec:authfailact}
The distributed design of $\borg$ reduces the likelihood of authentication failure. In rare cases where the UE cannot verify a BS (e.g., due to signer unavailability or missing authentication), it may continue scanning for alternative BSs. This behavior can be made configurable at the UE level, allowing users to prioritize either connectivity or security. For example, a UE may temporarily connect to an unauthenticated BS under constrained conditions while monitoring for a verifiable, authenticated BS to hand over to when available.


\subsubsection{Handling Mobility, Handover Protocol, and Roaming Scenario} The hierarchical and distributed design of $\borg$ enables efficient authentication in various mobility scenarios. During intra-gNB handovers, the UE can rely on previously authenticated $SIB$s, avoiding reauthentication. Newly issued $SIB$s remain verifiable using the core network’s PK embedded in the authenticated structure. Also, $\borg$’s threshold design allows uninterrupted authentication given the UE stays within range of any $t$ out of $n$ collaborating BSs, reducing handover latency and eliminating redundant reauthentication across adjacent BSs within the same domain. 
In roaming scenarios, authentication becomes more complex as the UE connects to a different network operator. $\borg$ assumes that the UE stores the $\pk_{\id_0}$ of its home network’s core-PKG within the USIM or eSIM. To authenticate $SIB$s from the serving network, the UE must obtain the public key of the serving network's core-PKG, distributed securely via non-3GPP access (e.g., Wi-Fi) and verified through a certificate signed by the home network’s core-PKG. eSIM technology enables this secure, dynamic credential provisioning. 
For the roaming scenario, our approach can leverage existing telecom-level PKI frameworks such as STIR/SHAKEN~\cite{wendt2019rfc}, mandated for U.S. carriers and based on operator-issued certificates for inter-operator authentication. These nationwide deployments demonstrate the practicality of implementing our solution within today’s cellular PKI ecosystem. Similar efforts are emerging globally, and 3GPP is developing complementary specifications for network PKI. 
While techniques, such as dynamic threshold adaptation or precomputed forward authentication shares, may further enhance handover efficiency across networks, we leave these optimizations to future work.

\subsubsection{Protection Against Relay Attacks}  
While $\borg$ enables UEs to authenticate $SIB$s and detect fake BSs, it does not prevent relay attacks, where an adversary rebroadcasts valid $SIB$s from a legitimate BS at higher signal strength to mislead UEs. Since relayed messages remain valid, verification alone is insufficient. Mitigating such attacks would require distance-bounding protocols, which demand substantial changes to cellular standards and hardware. 
To mitigate relay attacks without overhauling existing protocols, we extend $\borg$ with a time-bounded authentication mechanism. Each BS signs $SIB$s with a timestamp and short validity period, determined using transmission delay (reflecting propagation and processing time) and cryptographic delays (reflecting signing time) stored in a secure lookup table. The UE verifies message freshness by checking the timestamp against the current time, discarding any expired messages. $\borg$’s threshold design strengthens this defense by requiring timely inputs from multiple BSs, making it difficult for an attacker to relay a complete, valid message within the allowed window.

\vspace{-3mm}
\section{Performance Evaluation} \label{sec:PerformanceEvaluation}
This section evaluates and compares $\borg$ with alternative authentication schemes for 5G initial bootstrapping.

\subsection{Configuration and Experimental Setup}
\label{subsec:Configuration} 

\noindent\textbf{Hardware:} We assessed the efficiency of $\borg$ protocol utilizing a standard desktop equipped with an $12^{th}$ Gen Intel Core ${i7-12700H} @ 3.50~GHz$, $16~GiB$ RAM, a $512~GiB$ SSD, and Ubuntu $22.04.4~LTS$. The 5G testbed setup follows the configuration in Section~\ref{sec:feasibility}.  
Real network packets were investigated using the Network Signal Guru Android app installed on a OnePlus Nord 5G smartphone~\cite{netsigguru}.

\looseness-1 \noindent\textbf{Libraries:} We employed OpenSSL library\footnote{\url{https://openssl-library.org/}} for cryptographic primitives such as hash functions and EC operations (e.g., point multiplication, modular arithmetic), the Open Quantum-Safe library\footnote{\url{https://openquantumsafe.org/}} for NIST-PQC schemes, the Ringtail library\footnote{\url{https://github.com/daryakaviani/ringtail}} for the $\thpq$, and the blst\footnote{\url{https://github.com/Chia-Network/bls-signatures}} library for the $\textit{BLS}$ signature.

\noindent\textbf{Parameter Selection:} We configured the classical security level to $128$ bits, following NIST recommendations, and the post-quantum security to NIST Level I~\cite{alagic2022status}. NIST Level~I provides quantum resistance approximately equivalent to $128$-bit classical security. All cryptographic operations of the elliptic curve were performed over the standard curve \textit{secp224k1}, defined on a $224$-bit prime field, with \textit{SHA-256} used as the cryptographic hash function.

\noindent \textbf{srsRAN Configuration:} We observe that for the first SIB1 message, the srsRAN gNB utilizes $79$ bytes out of the allowed $372$ bytes. Accordingly, all subsequent evaluations report the computational and communication overhead for signing a $79$-byte SIB1 message, as reflected in the tables. Note that even considering slightly larger SIB1s observed from real networks, our results remain consistent.

\noindent\subsection{Evaluation Metrics} \label{subsec:EvaluationResult}  
Quantitative metrics include computational costs (e.g., signing, verification), 5G processing, cryptographic overhead (signature and key sizes), communication overhead (e.g., for transmission over-the-air),  and end-to-end (E2E) delay. The 5G processing latency captures the time network entities (e.g., gNBs and UEs) spend handling cryptographic material in packets, excluding cryptographic computations, and includes processing of signature fragments when applicable. Qualitative evaluation focuses on system architecture, accountability, and breach resiliency.

\noindent\subsection{Selection Criteria for Comparison Baselines} \label{subsec:SelectionCriteria} 
For PQ counterparts, we consider NIST-PQC signatures like lattice-based $\textit{ML\mbox{-}DSA}$ \cite{dang2024module}, $\textit{FN\mbox{-}DSA}$~\cite{soni2021falcon}, and hash-based $\textit{SLH\mbox{-}DSA}$ \cite{cooper2024stateless}. Given the very large signature and execution times of hash-based alternatives (e.g., $\textit{SLH\mbox{-}DSA}$ with a $7856$-byte signature, nearly $3\times$ that of $\textit{ML\mbox{-}DSA}$), we mainly focus on $\textit{ML\mbox{-}DSA}$ for direct comparison. While $\textit{FN\mbox{-}DSA-1024}$\cite{soni2021falcon} offers smaller signatures and similar performance, it relies on floating-point operations, making it less suitable for mobile platforms (UE). Even with efficient implementation, $\textit{FN\mbox{-}DSA}$ still incurs fragmentation, albeit less than $\textit{ML\mbox{-}DSA}$. Given $\textit{ML\mbox{-}DSA}$’s relative simplicity and prominence in the NIST PQC process, we adopt it as the main PQ baseline. Threshold variants of all PQC schemes are expected to incur significantly higher overhead, rendering them impractical for our use case.

\begin{table*}[ht]
\caption{Quantitative comparison of the alternative signature schemes in an overly ideal scenario (\textit{sending only signature}) for authenticating 5G cellular BSs.}
    \centering
    \resizebox{0.99\textwidth}{!}{
    \Large
    \renewcommand{\arraystretch}{1.2} 
    \begin{tabular} {|c||c|c|c|c|c|c|c|c|@{}c@{}|}
         \hline \textbf{Scheme} & \textbf{Sign (ms)} & \textbf{Verif. (ms)}  & \textbf{Packet Proc.(ms)} & \textbf{Transmission (ms)} &  \textbf{Crypto./Comm.(B)} & \textbf{PK (B)} &\textbf{E2E Delay (ms)}  \\\hline
          \textbf{$\textit{BLS}$} \cite{boneh2001short} & $0.42$ & $1.15$ & $0.03$ & $<0.01$ & $48/-$ & $96$ &  $1.60$\\ \hline
         \textbf{$\textit{ML\mbox{-}DSA}$} \cite{dang2024module} &  $0.12$ & $0.03$ & $0.57$ & $160.02$-$1280.02$ & $2420/2976$  & $1312$ & $160.74$-$1280.74$\\ \hline
        \textbf{$\textit{Schnorr\mbox{-}HIBS}^\dagger$} \cite{singla2021look} & $0.30$ & $1.27$ & $0.04$ & $<0.01$ & $64/-$ & $32$ & $1.61$ \\ \hline\hline
         \textbf{$\textit{Centralized\mbox{-}\borg}$} & $0.33$ & $1.27$ & $0.04$ & $<0.01$ &  $64/-$ & $32$ & $1.64$\\ \hline
         \textbf{$\textit{(2,3)\mbox{-}\borg}^\ast$} & $1.12$ & $1.27$ &  $0.04$ & $<0.01$ & $64/-$ & $32$ & $2.43$\\ \hline
         \textbf{$\textit{(2,3)\mbox{-}\borg}$} & $1.68$ & $1.27$ &  $0.04$ & $<0.01$ & $64/-$ & $32$ & $2.99$\\ \hline
    \end{tabular}}
\begin{tablenotes}
       \item {\small E2E delay represents the total time for signature generation, 5G delay (packet processing and transmission), and signature verification. A dash ($-$) indicates \textit{no additional} overhead, i.e., fits within the default SIB1 packet. $(2,3)\mbox{-}\borg^\ast$ considers a precomputed preprocessing phase. $^\dagger\textit{EC\mbox{-}Schnorr}$~\cite{schnorr1991efficient} shares identical timing and size metrics with $\textit{Schnorr\mbox{-}HIBS}$~\cite{singla2021look}.} 
    \end{tablenotes} 
    \label{tab:PerformanceComparison2}\vspace{-3mm}
\end{table*}

Given that $\textit{EC\mbox{-}Schnorr}$’s structure is inherently more amenable to threshold signing, supports practical implementation optimizations (e.g., \cite{maxwell2019simple}), and exhibits comparable timings and sizes to $\textit{ECDSA}$ \cite{johnson2001elliptic}, we adopt $\textit{EC\mbox{-}Schnorr}$ \cite{schnorr1991efficient} as the foundation for our scheme and the primary baseline for comparison. 
For conventionally secure signatures, we consider $\textit{EC\mbox{-}Schnorr}$~\cite{johnson2001elliptic} as an optimized standard, $\textit{BLS}$~\cite{boneh2001short} for its aggregation capabilities, and $\textit{Schnorr\mbox{-}HIBS}$~\cite{singla2021look} as a closely related certificateless scheme.  
Among closely related approaches, the scheme in~\cite{hussain2019insecure} utilizes a certificate chain with three distinct signature schemes, where our performance evaluation covers their architecture using standardized signature metrics. 
The work in~\cite{ross2024fixing} introduces a ``broadcast but verify" architecture using certificate-chain cryptography for 5G bootstrapping, proposing a separate \textit{signingSIB} message instead of embedding signatures in $SIB$s to improve efficiency.  Notably, $\borg$ can serve as a drop-in replacement in their design, enhancing both efficiency and security. $\textit{BARON}$~\cite{lotto2023baron}, a token-based protocol using symmetric encryption enables UE authentication but does not protect $SIB$s, allowing tampering without token invalidation, and is thus excluded from direct comparison. All baselines are evaluated against both centralized and threshold variants of $\borg$. 

\subsection{Experimental Results}
\looseness-1 Table~\ref{tab:PerformanceComparison2} presents a quantitative comparison of candidate schemes for signing a single SIB1 message, evaluating signing/verification time, 5G processing, cryptographic/communication overhead, and end-to-end (E2E) latency. Table~\ref{tab:PerformanceComparison3} extends this with both qualitative and quantitative analysis in the full 5G hierarchical bootstrapping context. We report the average results for 10000 iterations for all the schemes when not set up in the 5G testbed. As the 5G testbed requires manual intervention, we report the average of 10 iterations when the schemes are run on the testbed.

\subsubsection{Quantitative Comparison} 
This section presents the computational and communication overhead of $\borg$, accompanied by a comparison to alternative signatures. 

\noindent $\bullet$ \underline{\textit{Computational Costs:}} 
We begin by analyzing the signing and verification complexity for a single SIB1 message, with the results summarized in Table~\ref{tab:PerformanceComparison2}.  
%
$\textit{EC\mbox{-}Schnorr}$ and $\textit{BLS}$ exhibit low E2E delay, with $\textit{BLS}$ incurring slightly higher computational cost due to pairing-based operations. 
While $\textit{ML\mbox{-}DSA}$ achieves comparable execution times through optimized implementation, its large signature size results in substantial 5G communication overhead and a total delay of $1280.74$~ms, making it impractical for 5G SIB1 authentication, even without considering certificate hierarchy. $\textit{SLH\mbox{-}DSA}$, with a signature nearly three times larger and slower signing and verification ($11$~ms and $0.84$~ms, respectively), is even less suited for 5G authentication.  
$\textit{Schnorr\mbox{-}HIBS}$ and $\textit{Centralized$-$\borg}$ exhibit nearly identical execution time and communication overhead, resulting in comparable E2E delay for signing a single SIB1 message. In the threshold $\borg$ variant (e.g., $(2,3)$ configuration), signing takes approximately $1.12$~ms without preprocessing and $1.68$~ms with preprocessing included. Its verification time matches that of $\textit{Schnorr\mbox{-}HIBS}$ and $\textit{Centralized$-$\borg}$, which is particularly crucial for the resource-constrained UEs. Also, $\borg$ adopts Ringtail~\cite{boschini2024ringtail} as the $\thpq$ instantiation for distributed audit logging due to its performance benefits. Since $\thpq$ only signs the SIB1 signatures and does not affect BS/UE operations, it is excluded from the core performance comparison. Its preprocessing, signing, and verification take $89.4$, $3.29$, and $1.2$~ms, respectively.


In the full 5G evaluation (Table~\ref{tab:PerformanceComparison3}), which accounts for transmission of keys, certificates, and IDs, BS signing costs remain consistent with those reported in Table~\ref{tab:PerformanceComparison2}. For full verification, however, the UE must validate the SIB1 signature and, in certificate-based schemes, also verify certificates for the AMF and BS. This highlights the efficiency advantage of hierarchical schemes over flat alternatives such as $\textit{BLS}$, $\textit{EC\mbox{-}Schnorr}$, and $\textit{ML\mbox{-}DSA}$. While $\textit{BLS}$ benefits from signature aggregation, its pairing-based verification introduces considerable computational cost. Similarly, while $\textit{ML\mbox{-}DSA}$ offers relatively fast verification, its large key and signature sizes and high communication overhead result in a substantial E2E delay, rendering it infeasible for 5G bootstrapping, where SIB1 packets are sent every $20\sim160$~ms. 
Our $\textit{Centralized}\mbox{-}\borg$ achieves a total delay of $1.64$~ms, making it approximately $404\sim3221\times$ faster than $\textit{ML\mbox{-}DSA}$ ($662.47\sim{}5282.47$~ms). The $\textit{(2,3)\mbox{-}\borg}$, is slightly slower than the centralized version and $\textit{Schnorr\mbox{-}HIBS}$, yet remains significantly faster ($222\sim1767\times$) than $\textit{ML\mbox{-}DSA}$. Like $\textit{Schnorr\mbox{-}HIBS}$, both centralized and threshold $\borg$ variants transmit only $144$ bytes of cryptographic artifacts, fitting entirely within a single SIB1 packet, demonstrating the superior efficiency of our future-proof scheme over existing NIST-PQ solutions.

\begin{table*}[ht]
\caption{Comparison of candidate signature schemes for authenticating SIB1.}
    \centering
    \resizebox{0.99\textwidth}{!}{
    \Large
    \renewcommand{\arraystretch}{1.2} 
    \begin{tabular}{|@{}c@{}||@{}c@{}|@{}c@{}|@{}c@{}|@{}c@{}|@{}c@{}|@{}c@{}|@{}c@{}|} \hline 
    \multirow{2}{*}{\textbf{Scheme}} & \textbf{System Architecture} & \textbf{Sign} &\textbf{Full Verification}  & \textbf{5G} & \textbf{Crypto./Comm.} & \textbf{E2E} \\
    
    &\textbf{and Features}&\textbf{Delay (ms)}& \textbf{Delay (ms)}&\textbf{Delay (ms)}&\textbf{Overhead (B)}&\textbf{Delay (ms)}\\\hline
          \textbf{$\textit{BLS}$} \cite{boneh2001short} & 2-Level Certificate with Aggregation & $0.42$ & $3.46$  & $0.05$ & $240/-$ & $3.93$ \\ \hline
         \textbf{$\textit{EC\mbox{-}Schnorr}$} \cite{schnorr1991efficient} & 2-level Certificate & $0.30$ & $3.80$ & $0.05$  & $256/-$ & $4.15$\\ \hline
         \textbf{$\textit{ML\mbox{-}DSA}$} \cite{dang2024module} & 2-Level Certificate & $0.12$ & $0.12$ & $662.23\sim5282.23$ & $9884/12276$ & $662.47\sim5282.47$ \\ \hline
        \textbf{$\textit{Schnorr\mbox{-}HIBS}$} \cite{singla2021look} & Hierarchical& $0.30$ & $1.27$ & $0.04$ & $144/-$ & $1.61$\\ \hline\hline
    \textbf{$\textit{Centralized\mbox{-}\borg}$} & Hierarchical, Fail-Stop &$0.33$ & $1.27$ & $0.04$ & $144/-$ & $1.64$\\ \hline
         \textbf{$\textit{(2,3)\mbox{-}\borg}$} & Hierarchical, Fail-Stop, Threshold &$1.68$ & $1.27$ &  $0.04$ & $144/-$ & $2.99$\\ \hline
    \end{tabular}}
    \begin{tablenotes}
       \item {\small E2E delay presents the total time for signature generation, 5G delay and full verification. A dash ($-$) indicates \textit{no (additional)} overhead, i.e., fits within the default SIB1 packet.} 
    \end{tablenotes} 
    \label{tab:PerformanceComparison3} \vspace{-3mm}
\end{table*}

\noindent $\bullet$ \underline{\textit{Communication Overhead:}} 
In 5G BS authentication, certi- ficate-based schemes must transmit the SIB1 message, signature, public keys, and two certificates (for the AMF and BS) to establish key authenticity. In contrast, hierarchical schemes transmit only the SIB1 signature, corresponding public keys, and identities, eliminating certificates and reducing communication overhead.   
As shown in Table~\ref{tab:PerformanceComparison3}, flat schemes such as $\textit{BLS}$ and $\textit{EC\mbox{-}Schnorr}$ incur higher cryptographic overhead, while $\textit{ML\mbox{-}DSA}$ imposes substantial cryptographic and communication costs. Specifically, $\textit{ML\mbox{-}DSA}$ suffers from fragmentation, requiring $34$ network packets and incurring a total overhead of $12276$ bytes. In contrast, hierarchical schemes like $\textit{Schnorr\mbox{-}HIBS}$ and both centralized and threshold variants of $\borg$ maintain a compact $144$-byte overhead, fitting within a single SIB1 packet. 
While the threshold-$\borg$ requires inter-gNB communication for signature aggregation via the XnAP interface~\cite{etsi2} (using SCTP over IP~\cite{etsi3}), this delay is implementation-dependent and typically below $10$~ms~\cite{stewart2007stream}. Even under this upper bound, $\textit{(2,3)\mbox{-}\borg}$ remains significantly more efficient than NIST-PQ alternatives requiring fragmentation.

\noindent$\bullet$ \textit{\underline{UE Side Overhead}:} Due to the resource-constrained setting, it is imperative to investigate the computational overhead in the UE, posed by the schemes. We observe that all the schemes (including ours) except ML-DSA introduce only $0.015-0.02$ ms additional packet processing times in the UE. The potential reason for this is that UE only needs to process one SIB1 packet containing few additional bytes pertaining to the associated signature. In contrast, the packet processing time for ML-DSA on the UE side is much higher ($0.26$ms). However, note that UE still needs to verify the signature after processing the packet. For verification, we observe around $1.15-1.27$ms overhead for all the schemes (including ours) except ML-DSA. For verification, ML-DSA is faster with only $0.03$ms overhead. However, note that the primary concern with ML-DSA is its much larger communication overhead and the complexity it introduces in the protocol stack. Moreover, in the event of packet loss, BORG remains robust, as the signature can be delivered in the next broadcast of the SIB1 message without causing any packet-loss-related communication overhead. In contrast, NIST-PQC, like ML-DSA, requires synchronization from multiple consecutive SIB1 as discussed in section III. 

\subsubsection{Qualitative Comparison} 
This section presents a qualitative comparison of $\borg$ against related approaches, including certificate-based, hierarchical, and threshold schemes, in the context of 5G bootstrapping authentication.

\noindent $\bullet$ \underline{\textit{Limitations of Authentication with Flat Hierarchy:}} Direct use of non-hierarchical signatures for 5G BS authentication requires a certificate chain to authenticate public keys, increasing communication overhead and requiring the UE to verify both the SIB1 signature and the certificates for the AMF and BS keys. This adds considerable computational burden and is essential to mitigate threats like FBSs and MiTM attacks. As shown in Table~\ref{tab:PerformanceComparison3}, even efficient, conventionally secure schemes such as $\textit{EC\mbox{-}Schnorr}$ and $\textit{BLS}$ incur notable cryptographic overhead. For instance, if SIB1 configurations exceed $120$ bytes, $\textit{EC\mbox{-}Schnorr}$ requires fragmentation and delivery over two packets. Full-PQC schemes are even more demanding, with total communication costs nearing $12$~KB and requiring extensive fragmentation, posing serious reliability and availability issues. This makes PQC signatures not only computationally infeasible but also vulnerable to authentication failure if any signature fragment is lost in transit.

\looseness-1 \noindent $\bullet$ \underline{\textit{Limitations of Thresholding for 5G Authentication in the}} \underline{\textit{PQ Era:}} Threshold signatures, particularly those offering PQ guarantees, impose substantial overhead, making them impractical for 5G. Even conventionally secure threshold schemes, such as Schnorr-based signatures (e.g., $\textit{FROST}$~\cite{komlo2021frost}), which resemble $\textit{Schnorr\mbox{-}HIBS}$ and $(2,3)\mbox{-}\borg$, incur higher computational costs when applied to SIB1 authentication in a $(2,3)$ setting.   Several threshold variants of NIST-PQC signatures rely on resource-intensive techniques: multi-party computation (e.g., threshold-$\textit{Dilithi}$- $\textit{um}$ and threshold-$\textit{FN\mbox{-}DSA}$ reportedly require $12$~s and $6$~s to sign~\cite{cozzo2019sharing}), homomorphic hashing and commitments (e.g., $\textit{Dilizium}$~\cite{laud2022dilizium} incurs hundreds of milliseconds and $21120$-byte signatures), and fully homomorphic encryption~\cite{fu2021secure}, which leads to delays in the order of seconds. Additional constructions, such as those based on the Fiat–Shamir with Aborts paradigm~\cite{damgaard2022two} or hash-and-sign lattice approaches, suffer from transformation complexity and abort management overhead.  In contrast, $\borg$ offers a lightweight, practical alternative for 5G BS authentication. It maintains distributed trust and compromise resilience through thresholding while supporting FS security and PQ forgery detection, achieving a strong balance between security and deployability.

\noindent $\bullet$ \underline{\textit{PQ Assurances:}} As detailed in Section~\ref{sec:feasibility}, NIST-PQC signatures are currently unsuitable for 5G SIB1 authentication due to their large sizes, high computational costs, and substantial communication overhead. These signatures exceed the $SIB$ packet size limit. Experimental results (Tables~\ref{tab:PerformanceComparison2}-\ref{tab:PerformanceComparison3}) show that full-PQC schemes require excessive fragmentation and processing delays; for instance, $\textit{ML\mbox{-}DSA}$ requires $8$ additional packets for signature delivery, $12$ with the public key, and $33$ when including a 2-level certificate chain. 
Even without new message types, repeated SIB1 transmissions must carry fragmented signatures, further stressing the system. Broadcast unreliability and physical-layer constraints exacerbate deployment challenges. While conventional-secure schemes are efficient, they lack forgery detection and offer no protection against compromised BSs.  
$\borg$ achieves $3$ orders of magnitude faster execution and $85\times$ lower communication overhead than full-PQC alternatives like $\textit{ML\mbox{-}DSA}$, while offering distributed trust, forgery detection, and breach resiliency. Its efficiency, comparable to conventional hierarchical schemes, makes it a practical solution for 5G BS authentication.

\section{Related Work}\label{sec:RelatedWork}
We organize prior work on 5G BS authentication by technical approach, analyzing each category's contributions and limitations with respect to the design goals of \borg.

\noindent\textbf{PKI-Based BS Authentication.} Early proposals adapted PKI frameworks to protect SIB messages. Lee et al.~\cite{lee2009extended} and 
Zheng~\cite{zheng1996authentication} established certificate-based foundations for mobile network authentication. Hussain et al.~\cite{hussain2019insecure} provided the first systematic analysis of 5G bootstrapping vulnerabilities and proposed attaching signatures 
and certificate chains to SIB1 and SIB2 messages. Ross et al.~\cite{ross2024fixing} proposed a ``broadcast-but-verify'' model that transmits the signature in a separate \textit{signingSIB} message to decouple authentication overhead from SIB1. Gao et al.~\cite{gao2021evaluating} explored delegated signing to reduce per-BS computational cost, and Wuthier et al.~\cite{wuthier2025base} combined multi-factor authentication with offline blockchain-based certificate delivery. The 3GPP standardization body has explored PKI-based SIB protection in TR 33.809~\cite{FBSSpecifications}, though SIB1 remains unprotected in the current RRC specification~\cite{RRCSpec}. All PKI-based schemes share a fundamental limitation: certificate chains for AMF and BS 
keys routinely exceed the 372-byte SIB1 limit, requiring fragmentation across multiple packets, compounding verification cost on resource-constrained UEs, and remaining entirely vulnerable to quantum-capable adversaries.

\noindent\textbf{Token- and Symmetric-Based Schemes.} TESLA~\cite{perrig2003tesla} introduced lightweight broadcast authentication using a one-way key chain with timed key disclosure, relying solely on symmetric primitives at the cost of loose time synchronization. BARON~\cite{lotto2023baron} employed symmetric tokens for pre-authentication defense, introducing the concept of a Closed Trusted Entity (CTE) for secure connection initialization and handover in 5G 
networks. While these approaches avoid asymmetric certificate overhead, they do not protect SIB content itself: broadcast parameters remain susceptible to tampering without invalidating any token or key chain.

\noindent\textbf{Identity-Based and Certificate-Free Schemes.} To eliminate certificate chains, IBS-based schemes derive BS and AMF signing keys hierarchically from a master key pre-installed in the USIM, achieving compact overhead within a single SIB1 packet. Singla et al.~\cite{singla2021look} introduced \textit{Schnorr-HIBS}, a hierarchical IBS scheme that serves as the closest conventional-secure baseline to \borg. Ramadan et al.~\cite{ramadan2020identity} explored server-aided IBS 
to offload UE verification cost. Yu et al.~\cite{yu2024protecting} and Sun and Peng~\cite{sun20255g} further refined two-level HIBS constructions for LTE/5G. While IBS schemes achieve the best efficiency among conventional approaches, all rely on ECDLP hardness and are therefore broken by quantum-capable adversaries. Critically, none provides distributed trust, forgery detection, or any accountability mechanism following a BS compromise.

\noindent\textbf{Threshold and Distributed Authentication.} A small number of efforts have explored threshold signatures to distribute signing responsibility across multiple BSs. Sengupta and Lakshminarayanan~\cite{sengupta2024fast} explored online-offline threshold IBS for 5G IoT settings, and Vikhrova et al.~\cite{vikhrova2022multi} examined multi-SIM support scenarios involving coordinated BS authentication. However, these efforts operate under idealized assumptions and do not address the combined requirements of efficiency under 5G packet-size constraints, accountability upon BS compromise, or long-term breach resiliency. The broader challenge of designing a threshold authentication scheme that simultaneously meets 5G's strict overhead constraints and provides verifiable forgery detection remains unaddressed in prior work.

\noindent\textbf{Privacy and AKA-Layer Security.} Orthogonal to BS broadcast authentication, a line of work addresses privacy and mutual authentication at the NAS layer. Alnashwan et al.~\cite{alnashwan2024strong} presented a UC-secure authentication and handover protocol with strong user privacy guarantees, and Wang et 
al.~\cite{wang2021privacy} proposed encryption and KEM-based countermeasures targeting linkability vulnerabilities in 5G-AKA. These protocols operate after the bootstrapping phase and are orthogonal to SIB 
broadcast authentication; they assume a trustworthy BS connection has already been established, precisely the guarantee that \borg is designed to provide.

\noindent\textbf{Post-Quantum and Hybrid Solutions.} Recent work has begun addressing long-term quantum threats in cellular security. Efforts on PQ-AKA protocols~\cite{rossi2024enhancing, damir2022beyond, al2024post} apply lattice-based KEMs and PQ identification schemes to the NAS-layer authentication and key agreement procedures. Hybrid constructions combining NIST-PQC KEMs with symmetric primitives~\cite{vuppala2023post, ko20255g, scalise2024applied} aim to reduce the overhead of full PQC adoption for primary UE authentication. However, all of these efforts target unicast session establishment and are structurally incompatible with the one-to-many SIB broadcast 
authentication setting: KEMs establish shared secrets between two parties, and applying full PQC signatures to SIBs would compound fragmentation rather than resolve it, as demonstrated in Section~\ref{sec:feasibility}. None of these works addresses the unique constraints of initial BS bootstrapping authentication under 5G's strict size, timing, and broadcast requirements.

\vspace{-3mm}
\section{Limitations, Conclusion, and Future Work} 
\label{sec:conclusion} 
In conclusion, our feasibility assessment exposes the practical limitations of directly integrating NIST-PQC and conventional signature schemes into 5G bootstrapping authentication, primarily due to excessive signature sizes, certificate overhead, and fragmentation. To future-proof 
the 5G authentication, we propose $\borg$, a lightweight, distributed, and compromise-resilient framework that enables threshold authentication with forgery detection while meeting strict 5G constraints. Although $\borg$ does not provide full PQ security, it is designed to operate across all bootstrapping connections and currently secures critical $SIB$ messages. It is also worth noting that $\borg$ does not address privacy concerns related to UE-to-BS connections, passive eavesdropping, or other physical-layer threats, and remains vulnerable to overshadow attacks. Future work will extend protection to mitigate overshadow attacks in the PQ era. Looking ahead to 6G, we expect $\borg$ to help accelerate the adoption of 
secure and practical bootstrapping mechanisms beyond what 5G was able to incorporate.

\vspace{-3mm}
\section*{Acknowledgment} \label{sec:ack} \vspace{-1mm}
This work is supported by the CNS (2350213), NSF Grant No. 2112471, the University of Texas System Rising STARs Award (No. 40071109), and the startup funding from the University of Texas at Dallas. 
\vspace{-3mm}

\appendix
\section{APPENDIX} \label{sec:appendix}
\subsection{Acronyms}\label{subsec:acronmys}
A complete list of acronyms used throughout the paper is provided in Table~\ref{tab:notations}.

\begin{table}[ht]
\caption{List of Acronyms.}
    \centering
    \resizebox{0.99\linewidth}{!}{
    \small
    \renewcommand{\arraystretch}{1.2}
    \begin{tabular}{|l@{}|l@{}|}
        \hline
        \textbf{Acronyms} & \textbf{Description} \\ \hline \hline
        3GPP & Third Generation Partnership Project\\ \hline
        5G-RAN & 5G Radio Access Network \\ \hline
        5G-CN & 5G Core Network \\ \hline
        AMF & Access and Mobility Management Function\\ \hline
        AKA & Authentication and Key Agreement\\ \hline
        BS & Base Station\\ \hline
        CKG & Core Key Generator \\ \hline
        DL-SCH & Downlink Shared Channel\\ \hline
        ECDLP & Elliptic Curve Discrete Logarithm Problem\\ \hline
        \multirow{2}{*}{EUF-sID-CMIA} & Existential Unforgeability under a Selective-ID\\
        & adaptive Chosen Message-and-ID Attacks \\ \hline
        FBS & Fake Base Station\\ \hline
        FS & Fail-Stop\\ \hline
        gNB & Next Generation NodeB\\ \hline
        HIBS & Hierarchical Identity-Based Signature\\ \hline
        IBS & Identity-Based Signature\\ \hline
        ID & Identity of an entity (e.g., MAC address) \\ \hline
        IMSI & International Mobile Subscriber Identity \\ \hline
        LTE & Long Term Evolution\\ \hline
        MAC & Message Authentication Code\\ \hline
        MIB & Master Information Block\\ \hline
        MiTM & Man-in-The-Middle\\ \hline
        MPK & Master Public Key\\ \hline
        msk & Master Secret Key \\ \hline
        NAS & Non Access Stratum\\ \hline    
        NIST & National Institute of Standards and Technology\\ \hline
        PK & Public Key \\ \hline
        PKI & Public Key Infrastructure\\ \hline
        PLMN & Public Land Mobile Network\\ \hline
        PM & Post-Mortem \\ \hline
        PQC & Post-Quantum Cryptography \\\hline
        RRC & Radio Resource Control\\ \hline
        SDR & Software Defined Radio \\\hline
        SIB & System Information Block\\ \hline
        sk & Secret Key \\ \hline
        UE & User Equipment\\ \hline
        USIM & Universal Subscriber Identity Module\\ \hline
    \end{tabular}}
    \label{tab:notations}
\end{table}

\subsection{Signature Verification Correctness} \label{sec:verificationcorrectness} 
We demonstrate the correctness of the verification procedure by proving that an honestly generated signature in $\borg$ signature scheme satisfies the verification equation. 

Let $\sigma_{k,j} = (R_j, z_j)$ be a valid signature at hierarchy level $k$, generated according to the $\borg.\sign$ algorithm. Any verifier can validate this signature using the $\borg.\mverify$ procedure (Algorithm \ref{Alg:HITFSS3}), by computing the intermediate values $h_{\id_\ell}$ for each $\ell = 1, 2, \dots, k$, as well as $Q$ and $h_j$. 
The correctness of the signature is verified by checking whether the following equation holds:\vspace{-2mm}
\begin{equation*}
    g^{z_j} \stackrel{\mbox{?}}{=} R_j \cdot (Q \cdot Q_{\id_k} \cdot (\pk_{\id_0})^{\prod\limits_{\ell=1}^{k}h_{\id_\ell}})^{h_j} \mod p
\end{equation*}
\noindent which confirms the integrity and authenticity of the signed message under the Schnorr-based structure of the $\borg$ scheme. By Lagrange interpolation, the secret shares satisfy:\vspace{-2mm}
\begin{equation*}
    \sum_{i=1}^{\beta} \lambda_i \cdot \sk_{\id_{k,i}} = \sk_{\id_k}.    
\end{equation*}

\noindent Thus, if all signers behave honestly, then:\vspace{-2mm}
\begin{equation*}
    \scalebox{0.85}{$
    \begin{aligned}
    z_j = \sum_{i=1}^{\beta} z_{i,j} = \sum_{i=1}^{\beta} (d_{i,j} + e_{i,j} \cdot \rho_{i,j}) + h_j \cdot \sum_{i=1}^{\beta} \lambda_i \cdot \sk_{\id_{k,i}} \mod q.    
    \end{aligned}
    $}
\end{equation*}   
    
\noindent Hence, the left-hand side of the verification algorithm is:\vspace{-2mm}
\begin{equation*}
    g^{z_j} = g^{r_j} \cdot g^{h_j \cdot \sk_{\id_k}} = R_j \cdot g^{h_j \cdot \sk_{\id_k}} \mod p.    
\end{equation*}

\noindent Given the hierarchical key extraction procedure ensures:\vspace{-2mm}
\begin{equation*}
    g^{\sk_{\id_k}} = Q \cdot Q_{\id_k} \cdot \left( \pk_{\id_0} \right)^{\prod_{\ell=1}^{k} h_{\id_\ell}} \mod p,    
\end{equation*}
\noindent where:\vspace{-2mm}
\begin{equation*}
    Q \as \prod\limits_{\ell=1}^{k-1} (Q_{\id_\ell})^{\prod\limits_{\omega=\ell+1}^{k}h_{\id_\omega}}
\end{equation*}
Putting it together:\vspace{-2mm}
\begin{equation*}
    g^{z_j} = R_j \cdot \left( Q \cdot Q_{ID_k} \cdot PK_{ID_0}^{\prod_{\ell=1}^k h_{ID_\ell}} \right)^{h_j} \mod p.
\end{equation*}
This matches the verifier’s equation. Therefore, the verification algorithm accepts the signature.




\subsection{Security Proof} \label{subsec:securityproof} 

\noindent \textbf{Theorem 1.} 
\textit{If an adversary \A can $(q_{E}, q_P, q_S, q_{H_1,H_2})$-break $\borg$ in the random oracle model (Definition \ref{def:Experiment1}) with an advantage $\epsilon$ in time $\tau$ while having access to at most ${(t\mbox{--}1)}$-out-of-$n$ signing participants, where $q_E, q_P, q_S, q_{H_1}$, and $q_{H_2}$ denote queries to key extraction, preprocessing, signing, and hash functions $H_1$ and $H_2$, then an algorithm \C can be constructed to break the (EC)DLP in group $\mathbb{G}$.}

\begin{proof}\vspace{-1.5mm}
We assume that the adversary  \A is able to compromise ${t-1}$ signing participants by accessing the key extraction oracle  $\mathcal{O}_{E}$. It can also query the preprocessing oracle  $\mathcal{O}_{P}$, signing oracle  $\mathcal{O}_{S}$, and hash function oracles $\mathcal{O}_{H_1,H_2}$. We assume there are $t$ users in each level of the hierarchy. We note that the security proof can be generalized to $n$ users with $t$-out-of-$n$ thresholding. \A has control over $(t-1)$ signing participants. We consider a challenger \C which invokes \A as a black box and handles input and output queries, simulating the honest signing participant ($P_t$) across all queries and algorithms.  
By embedding a random challenge (in this case, a DLP instance) $\omega = g^{a^\ast} \in \mathbb{G}$ in query responses for the target $\id^\ast$ chosen by the forger. \A starts by picking a target identity ${\vec{\id}}^{\ast} = (\id^{\ast}_1, \dots, \id^{\ast}_\ell) \in {\mathbb{Z}_p}^\ell$ as the challenge identity.

\noindent $\bullet$ \ul{\textbf{Setup:}}  
Given $\kappa$, \C executes the $\borg.\setup(1^\kappa)$ (Algorithm \ref{Alg:HITFSS1}). It randomly selects ${\alpha_0 \asrand \mathbb{Z}_q}$, derives $\sk_{\id_0} \as H_1(\alpha_0)$, and computes ${\pk_{\id_0} \as g^{\sk_{\id_0}} \mod p}$. It then provides $(\pk_0, params)$ to \A while keeping $(\alpha_0, \sk_{\id_0})$ secret. 

\noindent $\bullet$ \ul{\textbf{Execute \A$^{\mathcal{O}_{\texttt{E}}, \mathcal{O}_{P}, \mathcal{O}_{S}}(\pk_{\id_0}, params)$:}} 
\A can adaptively issue a polynomially bounded number of queries, with \C acting as the honest party $P_t$ as follows. \vspace{-1mm}
\begin{itemize}[leftmargin=*]
    \item[-] \ul{\textit{Secret Key Extraction Oracle ($\mathcal{O}_{E}$):}} For a query~identity $\vec{\id} = (\id_1,\dots, \id_\ell) \in {\mathbb{Z}_p}^\ell$, where ${\vec{\id} \neq {\vec{\id}}^{\ast}}$ or any ${\id^{\ast}_i \not\in \vec{\id}}$ for ${i = 1, \dots, \ell}$, \C follows the $\borg.\keyextract(.)$ procedure and returns $(\{\sk_{\id_{\ell,i}}\}_{i=1}^{t},$ $\{\pk_{\id_{\ell,i}}\}_{i=1}^{t}, \vec{Q}_{\id_\ell})$.  
    For $\vec{\id} = {\vec{\id}}^{\ast}$, \C embeds the challenge $\omega \as g^{a^\ast} \mod p$ by
    setting $\pk_{\id_{\ell,t}} = \omega$. It then derives the secret and public keys for the remaining ${t-1}$ participants by choosing $a_i \asrand \mathbb{Z}_q$ and computing $\pk_{\id_{\ell,i}} \as \omega^{\lambda_t} \cdot g^{\sum\limits_{i=1}^{t-1} {\lambda_i \cdot a_{\ell,i}}}$ for ${i=1,\dots,t-1}$. The group verification key $Q_{\id^\ast}$ and $\Vec{Q}_{\id^\ast}$, follows the same procedure as the $\borg.\keyextract(.)$. \vspace{-1mm}

    \item[-] \ul{\textit{Preprocessing Oracle ($\mathcal{O}_{P}$):}} Given $J$ and the ID of the signing participants (where $\vec{\id} \neq {\vec{\id}}^{\ast}$ or any $\id^{\ast}_i \not\in \vec{\id}$ for ${i = 1, \dots, \ell}$), it returns the commitment value of the participants ($\mathcal{L}_{i} \as (i, \{E_{i,j}\}_{j=1}^{J},$ $\{D_{i,j}\}_{j=1}^{J})$) for ${i=1,\dots,t-1}$, following $\borg.\preprocess(.)$. \vspace{-1mm}
    
    \item[-] \ul{\textit{Signing Oracle ($\mathcal{O}_{S}$):}}  For a message $m$ and $\vec{\id} = (\id_1,\dots, \id_\ell) \in {\mathbb{Z}_p}^\ell$, where $\vec{\id} \neq {\vec{\id}}^{\ast}$ or any $\id^{\ast}_i \not\in \vec{\id}$ for ${i = 1, \dots, \ell}$, \c responds by using the $\borg.\keyextract(.)$ algorithm to extract the secret key, obtain the commitment values $\borg.\preprocess(.)$, and produce a valid signature relying on the signing algorithm $\borg.\sign(.)$. The signature queries on the target $\id^\ast$ would be rejected. \vspace{-1mm}
    
    \item[-] \ul{\textit{Hashing with a Random Oracle ($\mathcal{O}_{H_1}, \mathcal{O}_{H_2}$):}} Let $\mathcal{L}_{H_1}$ and $\mathcal{L}_{H_2}$ denote the query logs for the hash functions $H_1$ and $H_2$, respectively. Upon receiving a query to $H_1$ on the tuple $(\id_\ell, \vec{Q}_{\id_\ell})$, the challenger \C checks $\mathcal{L}_{H_1}$: if an entry exists, it returns the stored value; otherwise, it samples $x \asrand \mathbb{Z}_q$, stores it in $\mathcal{L}_{H_1}$, and returns $x$. Similarly, for a query to $H_2$ on $(R_j, Q_{\id_\ell}, m_j)$, \C checks $\mathcal{L}_{H_2}$ and either returns the stored value or samples $x' \asrand \mathbb{Z}_q$, stores it, and returns $x'$. \vspace{-1mm}
\end{itemize}

\noindent$\bullet$ \ul{\textit{\textbf{Forgery of \A:}}} \A produces a valid signature ($m^\ast, \sigma^\ast$) for $\id^\ast$ under the master and group public keys ($\pk_{\id_0}, Q_{\id_k}$).  
\A wins the experiment if satisfies the conditions mentioned in Definition \ref{def:Experiment1}.

\noindent$\bullet$ \ul{\textit{\textbf{Solution to DLP via $\mathcal{C}(\omega)$:}}} Given the adversary \A with access to $t-1$ signing participants can produce a signature forgery $\sigma^\ast$, \C can uses \A as a black-box forger, and utilize the generalized forking lemma, solves the DLP for the embedded challenge $\omega$, as shown below: \vspace{-1mm}
\begin{itemize}[leftmargin=*]
    \item Applying GFL, the adversary \A is run twice with the same random tape while having different hash queries to $\mathcal{O}_{H_2}$, then, \C can obtain two signature forgeries $\sigma^\ast = ({R_j}^\ast, {z_j}^\ast), {\sigma^\ast}'=({R_j}^\ast, {{z_j}^\ast}')$.  \vspace{-1mm}
    
    \item Given the query responses from preprocessing and $\mathcal{H}_{1}$ are the same, we get to the following two equations:\vspace{-2mm} 
    \begin{equation*}
    \scalebox{0.85}{$
    \begin{aligned}
        \hspace{-3mm}\sum\limits_{i=1}^{t}{z^\ast}_{i,j}= \sum\limits_{i=1}^{t}{d^\ast}_{i,j} + \sum\limits_{i=1}^{t} {e^\ast}_{i,j}\cdot {\rho^\ast}_{i,j} + \sum\limits_{i=1}^{t} \lambda_i\cdot \sk_{\id_{i,j}}.{h^\ast}_j \mod q
    \end{aligned}
    $}
    \end{equation*}\vspace{-3mm}
    \begin{equation*}
    \scalebox{0.85}{$
    \begin{aligned}
    \hspace{-3mm}\sum\limits_{i=1}^{t}{z^\ast}'_{i,j}= \sum\limits_{i=1}^{t}{d^\ast}_{i,j} + \sum\limits_{i=1}^{t} {e^\ast}_{i,j}\cdot {\rho^\ast}_{i,j} + \sum\limits_{i=1}^{t} \lambda_i\cdot \sk_{\id_{i,j}}.{h^\ast}'_j \mod q
    \end{aligned}
    $}
\end{equation*}\vspace{-3mm}
    \item From the above equations, we get:\vspace{-3mm}
     \begin{equation*}
    \scalebox{0.85}{$
    \begin{aligned}
        \sk_{\id_\ell} = \frac{\sum\limits_{i=1}^{t}({z^\ast}_{i,j}-{z^\ast}'_{i,j})}{{h^\ast}_j-{h^\ast}'_j} 
    \end{aligned}
    $}
    \end{equation*}\vspace{-3mm}
    \begin{equation*}
    \scalebox{0.85}{$
    \begin{aligned}
        a^\ast = \frac{1}{\lambda_t} \times (\frac{\sum\limits_{i=1}^{t}({z^\ast}_{i,j}-{z^\ast}'_{i,j})}{{h^\ast}_j-{h^\ast}'_j} - \sum\limits_{i=1}^{t-1} \lambda_i \cdot \sk_{\id_\ell})
    \end{aligned}
    $}
    \end{equation*}
    \item Finally, \C obtains the DLP of $\omega$ in $\mathbb{G}$. 
\end{itemize}
\vspace{-5mm}
\end{proof}

\noindent \textbf{Theorem 2.} 
\textit{$\borg$ provides $\lambda_1$-bit signer-side fail-stop security against quantum-capable adversaries controlling up to ${(t-1)}$-out-of-$n$ signing participants, as formalized in Definition~\ref{def:Experiment2}, and $\lambda_2$-bit non-repudiation security against quantum adversaries with access to one-out-of$t$ signing participants, as captured in Definition~\ref{def:Experiment3}. Both guarantees rely on the hardness of breaking the second preimage resistance of a cryptographically secure hash function, while providing $\kappa$ bit verifier-side security via EUF-sID-CMIA in the random oracle model as defined in Definition~\ref{def:Experiment1}.}



%
%
\vspace{-2mm}
\begin{proof}
Let $\mathcal{A}$ be a quantum-capable adversary with access to the signing oracle, preprocessing oracle, and control over $({t-1})$ out of $n$ signing participants. We assume $t$ users per level in the hierarchy, though the proof generalizes to $t$-out-of-$n$ thresholding. 
Assume, for contradiction, that \A succeeds in the signer-side fail-stop experiment (Definition~\ref{def:Experiment2}) with non-negligible advantage $\epsilon$, producing a forged signature $\sigma^\ast_k$ on a message $m^\ast$ that (i) passes verification ($1 \as \borg.\mverify(m^\ast, \vec{\id}_k, \vec{Q}_{\id_k}, \sigma^\ast_k)$), and (ii) cannot be proven invalid via $0 \as \borg.\pof(\{\hat{e}^\ast_{i,\mathrm{j}}\}_{i=1}^{t}, \{\hat{d}^\ast_{i,\mathrm{j}}\}_{i=1}^{t}, m^\ast,$ $\sigma^\ast_k, hist)$, i.e., honest signers fail to produce a valid forgery proof. 
We construct a reduction algorithm $\mathcal{C}_1$ that uses \A as a black box, handles queries, and simulates the honest signing participant ($P_t$) across all queries and algorithms. For the target message $m^\ast$, $\mathcal{C}_1$ embeds a second preimage challenge as a commitment from an honest signer by programming commitment values (${e^\ast_{i,t}\as H_1(\hat{e}^\ast_{t,j}||j||\id_{k,t}), d^\ast_{i,t}\as H_1(\hat{d}^\ast_{t,j}}$ ${||j||\id_{k,t})}$) where the random nonces are ${(\hat{e}^\ast_{i,t}, \hat{d}^\ast_{i,t})\asrand \zq\times\zq}$.  

\looseness-1
Since forgery detection in $\borg.\pof$ depends on these hash-based commitments that derive the shared component $R_j$, any successful forgery must reproduce these commitments without access to the original random nonces. Thus, if \A outputs a successful forgery $\sigma^\ast$ and corresponding alternate preimage $(\hat{e}_{i,t}^{*\prime}, \hat{d}_{i,t}^{*\prime})$ that satisfy the conditions mentioned in Definition~\ref{def:Experiment2} such that $H_1(\hat{e}_{i,t}^{*\prime}||j||\id_{k,t}) = H_1(\hat{e}_{i,t}^{*}||j||\id_{k,t})$ and $H_1(\hat{d}_{i,t}^{*\prime}||j||\id_{k,t}) = H_1(\hat{d}_{i,t}^{*}||j||\id_{k,t})$ while ${(\hat{e}_{i,t}^{*\prime}, \hat{d}_{i,t}^{*\prime}) \neq (\hat{e}_{i,t}^{*}, \hat{d}_{i,t}^{*\prime})}$, then $\mathcal{C}_1$ outputs $(\hat{e}_{i,t}^{*\prime}, \hat{d}_{i,t}^{*\prime})$ as a valid second preimage for the embedded challenge. This contradicts the assumed hardness of second preimage resistance of $H_1$. 
Hence, under Grover’s algorithm \cite{grover1996fast}, the quantum adversary's success probability is reduced to $\mathscr{O}(2^n)$, yielding $\lambda_1$-bit PQ security for an $n$-bit hash function.
\end{proof}

%
%
\vspace{-2mm}
Based on Definition~\ref{def:Experiment3}, we prove fail-stop non-repudia- tion of $\borg$ under the hardness of second preimage resistance of a cryptographically secure hash function.\vspace{-2mm}

\begin{proof}
Let $\mathcal{A}$ be a quantum-capable adversary with access to preprocessing queries and control over one of the $t$ signing participants. Suppose $\mathcal{A}$ wins the non-repudiation experiment in Definition~\ref{def:Experiment3} with non-negligible probability $\epsilon$; that is, $\mathcal{A}$ participates in generating a valid signature $\sigma_k$ such that (i) it passes verification $1 \as \borg.\mverify(m, \vec{\id}_k, {\vec{Q}}_{\id_k}, \sigma_k)$, and (ii) later constructs a forged proof $\pi^\ast$ satisfying ${1 \as \borg.\fverify(\alpha_{k},}$ ${\sk_{\id_{k-1}}, \vec{Q}_{\id_k}, m, \sigma_k', \pi^\ast)}$, despite $\sigma_k$ being honestly generated according to the signing protocol. We construct a reduction algorithm $\mathcal{C}_2$ that treats $\mathcal{A}$ as a black box, simulates all $t$ signers, and handles all input/output queries. For the message $m$, $\mathcal{C}_2$ embeds second preimage challenges as the commitments of all $t$ signers by programming values ($e^\ast_{i,j}\as H_1(\hat{e}^\ast_{i,j}||j||\id_{k,i}), d^\ast_{i,j}\as H_1(\hat{d}^\ast_{i,j}||j||\id_{k,i}$) where $(\hat{e}^\ast_{i,j}, \hat{d}^\ast_{i,j})\asrand \zq\times\zq$ for all $i=1,\dots,t$.

\looseness-1
Since forgery detection in $\borg.\pof$ relies on hash-based commitments used to derive the shared component $R_j$, a valid forgery must reproduce these commitments without access to the original nonces. 
For the $t$ signers to falsely prove a legitimate signature $\sigma_k$ as a forgery, at least one signer must find a distinct preimage $\hat{e}_{i,j}^{*\prime} \neq \hat{e}^\ast_{i,j}$ such that $H_1(\hat{e}_{i,j}^{*\prime}||j||\id_{k,i}) = H_1(\hat{e}^\ast_{i,j}||j||\id_{k,i})$, for some ${i\in\{1,\dots,t\}}$. If the quantum-capable adversary $\mathcal{A}$ outputs a valid-looking proof of forgery $\pi^\ast$ using such alternate preimage $(\hat{e}_{i,j}^{*\prime}, \hat{d}_{i,j}^{*\prime})$ that satisfies the conditions mentioned in Definition~\ref{def:Experiment3}, then the reduction $\mathcal{C}_2$ extracts a second preimage for the embedded challenge, contradicting the assumed hardness of second preimage resistance of $H_1$. Under Grover’s algorithm~\cite{grover1996fast}, the adversary’s success probability reduces to $\mathscr{O}(2^{n/2})$, yielding $\lambda_2$-bit PQ security for an $n$-bit hash function.
\end{proof}

\bibliographystyle{elsarticle-num}
\bibliography{SalehRef}             

\end{document}